\documentclass[runningheads]{llncs}
\usepackage{graphicx}
\usepackage{amsmath}
\usepackage{amssymb}
\usepackage{xcolor}
\usepackage{mathtools}
\usepackage{hyperref}

\begin{document}
\title{Object-Centric Alignments}
%
%
\author{Lukas Liss\and
Jan Niklas Adams \and
Wil M.P. van der Aalst}
\authorrunning{L. Liss et al.}
%
\institute{Chair of Process and Data Science\\
RWTH Aachen University,
Aachen, Germany\\
\email{lukas.liss@rwth-aachen.de, \{niklas.adams,wvdaalst\}@pads.rwth-aachen.de}\\}
\maketitle              
\begin{abstract}
Processes tend to interact with other processes and operate on various objects of different types.
These objects can influence each other creating dependencies between sub-processes.
Analyzing the conformance of such complex processes challenges traditional conformance-checking approaches because they assume a single-case identifier for a process.
To create a single-case identifier one has to flatten complex processes.
This leads to information loss when separating the processes that interact on some objects.
This paper introduces an alignment approach that operates directly on these object-centric processes.
We introduce alignments that can give behavior-based insights into how closely related the event data generated by a process and the behavior specified by an object-centric Petri net are.
The contributions of this paper include a definition for object-centric alignments, an algorithm to compute them, a publicly available implementation, and a qualitative and quantitative evaluation.
The qualitative evaluation shows that object-centric alignments can give better insights into object-centric processes because they correctly consider inter-object dependencies.
Findings from the quantitative evaluation show that the run-time grows exponentially with the number of objects, the length of the process execution, and the cost of the alignment.
The evaluation results motivate future research to improve the run-time and make object-centric alignments more applicable for larger processes.

\end{abstract}
\section{Introduction}
\label{sec:Introduction}
Process mining provides insights into processes by analyzing event data generated by these processes.
When analyzing a process, one standard pipeline consists of extracting data, discovering a process model, and checking the conformance of the process with specifications~\cite{carmona_conformance_2018}.
This paper presents an approach to compute alignments to check the conformance of object-centric processes for which traditional conformance-checking methods fail to give correct insights.

Traditional process mining approaches depend on the assumption that a process is defined by a single case notion meaning that all actions created for one object define a process execution.
Processes in the real world tend not to fit that assumption.
An example of that is a typical supply chain process.
Supply processes are happening on raw materials, production processes on raw materials and products, shipping processes on products and orders, and payment processes operating on orders and customers.
One execution of the supply chain is not defined by a single object.
Real-world processes consist of multiple sub-processes operating on multiple object instances from various types.
These sub-processes can have synchronization points and long-term dependencies between different objects.

\begin{figure}[t]
    \centering
    \includegraphics[width=0.7\textwidth]{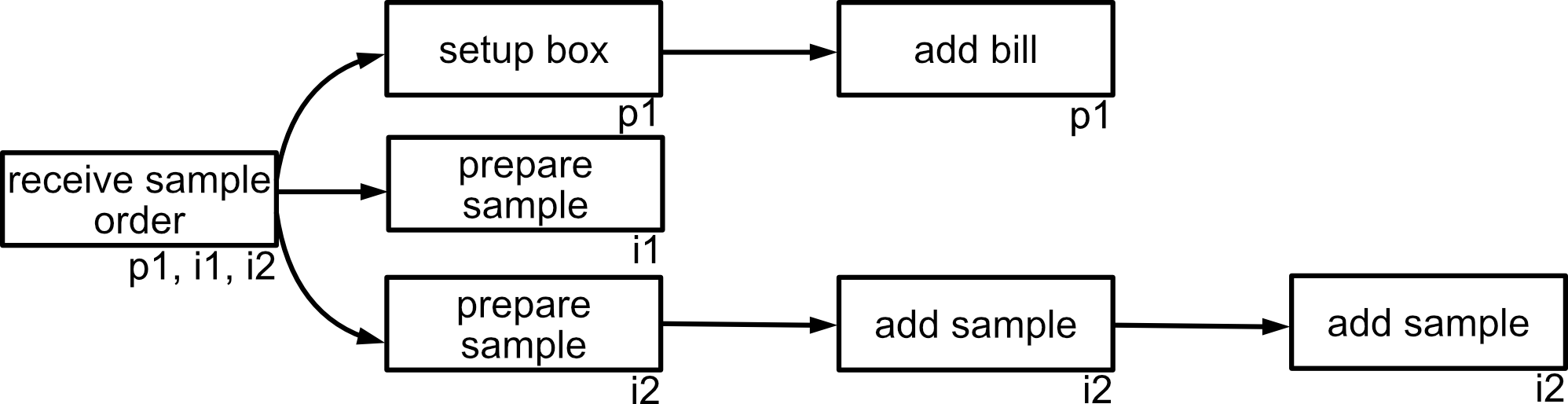}
    \caption{A process execution of our running examples. Events are associated with objects of type package (prefix p) or item (prefix i). The process execution describes the partial order of events induced by the individual objects.}
    \label{fig:px}
\end{figure}

\begin{figure}[t]
    \centering
    \includegraphics[width=\textwidth]{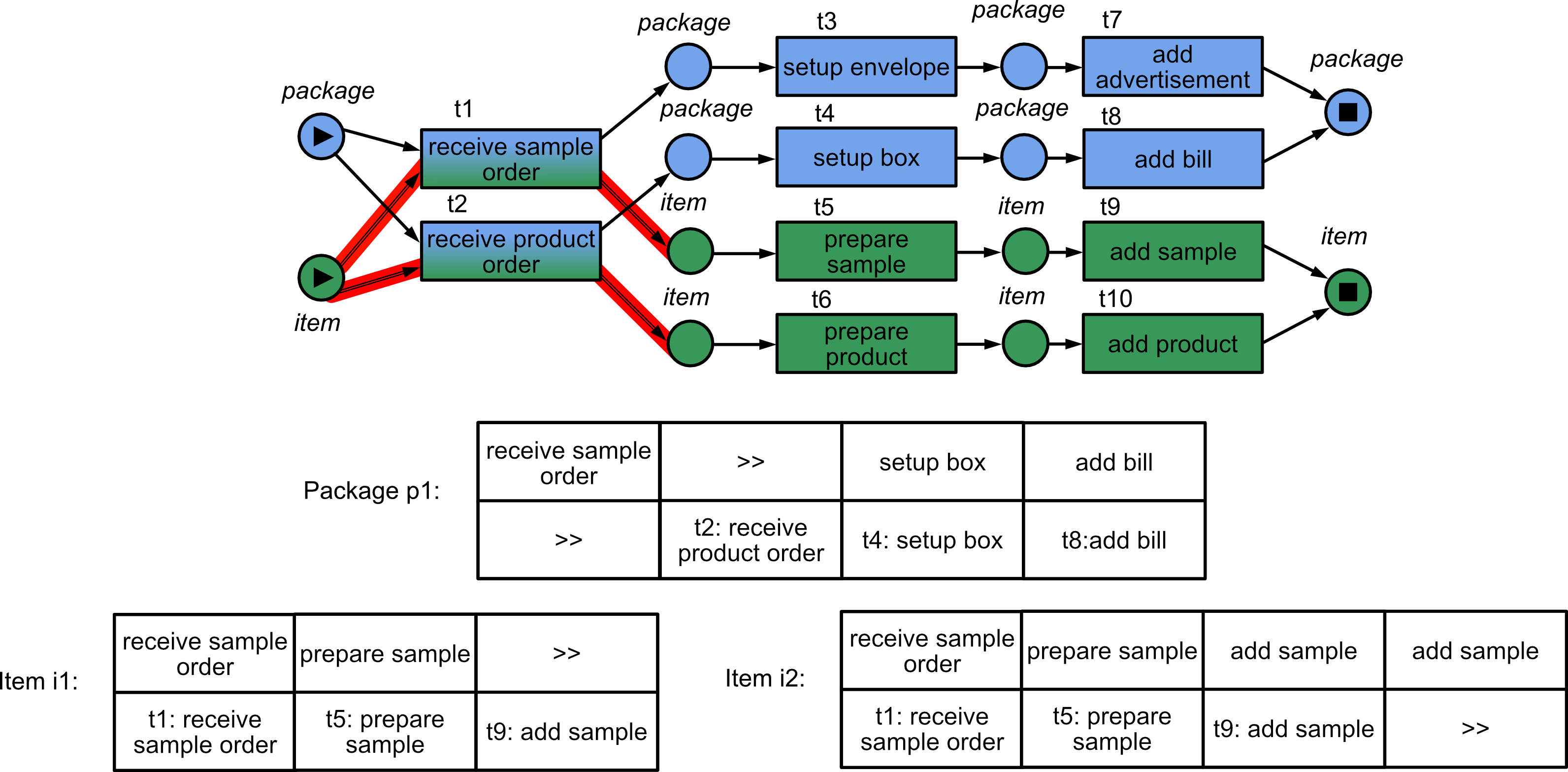}
    \caption{Top: De-jure model as an object-centric Petri net. Variable arcs in the model are marked red. Bottom: The process execution of \autoref{fig:px} is flattened to the individual objects and aligned to the de-jure model's subnet of the object's types.}
    \label{fig:pn-w-flattened}
\end{figure}

Recently, approaches have been made to generalize the notion of a process so that one can describe these complex processes.
Those approaches fall under the umbrella term of \textit{object-centric process mining}~\cite{DBLP:books/sp/22/Fahland22}.
Object-centric process mining generalizes traditional process mining techniques such that one does not follow one object through one process, but multiple objects through multiple, connected sub-processes.
This generalization increases the complexity.
In traditional process mining, one execution of a process is called \textit{case} and is defined by a sequence of events~\cite{van_der_aalst_process_2012}.
An object-centric \textit{process execution} is a graph describing the partial order between events in different sub-processes~\cite{adams_defining_2022}.
So far, a fitness notion based on replay has been proposed for object-centric conformance checking, but the approach is constrained to replayable behavior\cite{adams_precision_2021}. However, process owners are typically interested in aligning observed behavior to modeled behavior, finding deviations in their execution of the process, i.e., using alignments~\cite{adriansyah_aligning_2014}. The notion, calculation, implementation, and feasibility analysis for alignments on object-centric process mining are, so far, missing.

The running example used in this paper describes a packaging process with cross-object dependencies. In our example, one process execution refers to multiple items and one package. \autoref{fig:px} depicts a process execution as a graph of events describing the partial order between item and package objects.
\autoref{fig:pn-w-flattened} shows an object-centric Petri net~\cite{aalst_discovering_2020} of the packaging process.
An object-centric Petri net differs from traditional Petri nets by introducing place types and variable arcs. Tokens are typed and places of one type can only hold tokens of the same type. Variable arcs can consume an arbitrary amount of tokens.
Places are typed according to either item (green) or package (blue). We, furthermore, colored transitions with the colors of the types that are involved in this transition for clarity and highlighted the variable arcs with red.
In the given process, the path of the package depends on the path of items and vice-versa.
The package and the items can either be part of a sample order or a product order, but never both, depending on whether \textit{receive sample order} or \textit{receive product order} fires. Although the objects do not share any more events after the start, their allowed behavior still depends on whether they belong to a sample or product order.

If a process owner would like to find deviations in their object-centric processes today, they would need to \textit{flatten}~\cite{van_der_aalst_object-centric_2019} the observed process executions and apply traditional alignments to the object-centric Petri net's subnets of the same type.
We show this for our example process execution of \autoref{fig:px} in \autoref{fig:pn-w-flattened}.
If flattened to one trace per object, the three resulting traces get aligned to the type's subnet in a way that is not possible in the composed model.
As mentioned, activity \textit{receive sample order} and \textit{receive product order} can never happen both in one process execution since the de-jure model forces a decision between product orders and sample orders.
But the flattened alignments do not agree on which activity should happen.
The alignment for \textit{p1} has \textit{receive product order} in the model part whereas the alignments for \textit{i1} and \textit{i2} have \textit{receive sample order} in their model part.
As shown by the running example, computing alignments on object-centric processes requires more than just finding alignments for each object individually.
Aligning the sub-processes for all objects by respecting their object dependencies creates a computationally complex problem that we tackle in this paper.

This paper presents four contributions to enable and investigate object-centric alignments.
First, we generalize the notion of an alignment to object-centric processes.
Second, we present an algorithm to compute optimal object-centric alignments.
Third, we implemented our algorithm and make it publicly accessible as an open-source project\footnote[1]{https://github.com/LukasLiss/object-centric-alignments} based on the open-source object-centric process mining library \textsc{opca}~\cite{adams_ocpa_2022}.
Fourth, we evaluate the quality and the computation time of object-centric alignments on real-world event data.
Thereby, we gain insights into the scalability and suitability of the approach.

This paper is structured in the following way.
We present related work in \autoref{sec:Related} and preliminaries in \autoref{sec:Preliminaries}.
Then, we define object-centric alignments in \autoref{sec:alignment}.
Our algorithm to compute alignments consists of two parts:
constructing the synchronous product net (\autoref{sec:sny-prod-net}) and finding an optimal alignment in the synchronous product net (\autoref{sec:alignm-from-syn-prod-net}).
We present a qualitative and quantitative evaluation in \autoref{sec:evaluation} and conclude the paper in \autoref{sec:conclusion}.

\section{Related Work}
\label{sec:Related}
Process mining includes discovery, conformance checking, and enhancement of business processes \cite{van_der_aalst_process_2012}.
Our approach belongs to the category of conformance checking, where behavior from the event log is compared to allowed behavior that is specified by a de-jure model \cite{carmona_conformance_2018}.
For traditional processes, there exists a variety of conformance checking approaches \cite{dunzer_conformance_2019}.
The majority of them use either a token-based replay \cite{rozinat_conformance_2008} approach, or an alignment \cite{adriansyah_aligning_2014} approach.
Both have been used to derive quality metrics like precision \cite{adriansyah_alignment_2013} and fitness \cite{adriansyah_cost-based_2011}.
Unlike token-based replay, alignments are independent of the structure of the de-jure model \cite{adriansyah_aligning_2014}.
Like our calculation, the traditional alignment calculation defined by Adriansyah et al. uses a two-step algorithm to compute alignments \cite{adriansyah_aligning_2014}.
Adriansyah et al.'s approach creates a synchronous product net such that finding optimal alignments relates to finding a shortest path in that net.
This is a well-studied problem that can be solved with the Dijkstra \cite{dijkstra_note_1959} or $A^*$ \cite{dechter_generalized_1985} algorithm.
Different ways to speed up the calculation have been researched \cite{van_dongen_efficiently_2018}\cite{song_efficient_2017}.
However, the alignment algorithm assumes the process to have a single case identifier and can therefore not be used for compositions of processes that operate on multiple objects.

Multiple extensions to the traditional alignment algorithm have been made over the years that consider additional dimensions together with the workflow dimension \cite{borrego_conformance_2014}\cite{burattin_conformance_2016}.
Thereby they use higher-order nets to represent the additional dimensions.
The data and resource-aware conformance checking approach from de Leoni et al. uses data Petri nets \cite{de_leoni_data-_2012}.
Felli et al. use data Petri nets together with satisfiability modulo theories to compute data-aware alignments \cite{DBLP:conf/bpm/FelliGMRW21}.
Sommers et al. constructed a $\nu$-Petri net to calculate resource-constrained alignments \cite{bernardinello_aligning_2022}.
But all of the approaches above assume the process to have a single-case identifier.

There are approaches that lift this generalization and model processes as interacting sub-processes.
Multi-agent process models describe the behavior of agents and their interaction by composing Petri nets \cite{nesterov_discovering_2022}.
Object-centric process mining \cite{van_der_aalst_object-centric_2019} extends the notion of processes so that they can interact and operate on objects from different types.
Adams et al. defined the notion of cases and variants for object-centric processes \cite{adams_defining_2022} as event graphs instead of event sequences.
The defined process executions serve as input for our alignment calculation.
The other data type that we use as input is an object-centric Petri net \cite{aalst_discovering_2020} that can describe allowed behavior.
Object-centric Petri nets can be discovered from object-centric event logs using the discovery algorithm from van der Aalst and Berti \cite{aalst_discovering_2020}.
Precision and fitness metrics to evaluate the quality of a model have, recently, been proposed~\cite{adams_precision_2021}.
However, techniques to check conformance to a de-jure model and spot deviations, such as object-centric alignments, are so far missing.

\section{Preliminaries}
\label{sec:Preliminaries}
Object-centric process mining deals with events that operate on a variety of objects of different types.
Events are activities that happen at a timestamp for a number of objects of different types.
$\mathbb{U}_{event}$ is the Universe of event identifiers.
The universe $\mathbb{U}_{act}$ contains all visible activities.
$\mathbb{U}_{typ}$ is the universe of all object types.
The universe of objects is $\mathbb{U}_{obj}$.
Each object has exactly one type associated with it
$\pi_{type} : \mathbb{U}_{obj} \to \mathbb{U}_{typ}$.
$\mathbb{U}_{time}$ is the universe of all timestamps.

\begin{definition}[Event Log]
$L = (E, O, OT, \pi_{act}, \pi_{obj}, \pi_{time}, \pi_{trace})$ is an event log with:

\begin{itemize}
    \item $E \subseteq \mathbb{U}_{event}$ is a set of events, $O \subseteq \mathbb{U}_{obj}$ is a set of objects,
    \item $OT = \{ \pi_{type}(o) | o \in O \}$ is a set of object types,
    \item $\pi_{act}: E \to \mathbb{U}_{act}$ maps each event to an activity,
    \item $\pi_{obj}: E \to \mathcal{P} (\mathbb{U}_{obj}) \setminus \{ \emptyset \}$ maps each event to at least one object,
    \item $\pi_{time}: E \to \mathbb{U}_{time}$ maps each event to a timestamp, and
    \item $\pi_{trace}: O \to E^*$ maps each object onto a sequence of events such that
    $\pi_{trace}(o) = \langle e_1, ..., e_n \rangle$ with\\
    $\{ e_1, ..., e_n \} = \{ e \in E| o \in \pi_{obj}(e)\}$ and $\forall_{i \in \{1, ..., n-1\}}\; \pi_{time}(e_i) \leq \pi_{time}(e_{i+1})$ 
\end{itemize}
\end{definition}

Event logs can contain events from multiple process executions. When analyzing the behavior we want to extract one process execution.

\begin{definition}[Process Execution]
Let $L = (E, O, OT, \pi_{act}, \pi_{obj}, \pi_{time}, \pi_{trace})$ be an object-centric event log.
The object graph $OG_{L} {=} (O,I)$ with $I {=} \{ \{o, o'\} |\allowbreak \exists_{e \in E} \{o, o' \} \subseteq \pi_{obj}(e) \wedge o \neq o'\}$ connects objects that share events.
The connected components $ con(L) = \{ X \subseteq O | X \text{ is a connected component in } OG_L \}$ of the object graph are sets of inter-dependent objects.
Each set $X \in con(L)$ defines a process execution of $L$.
A process execution is a graph $P_X = (E_X, D_X)$ with nodes $E_X = \{ e \in E | X \cap \pi_{obj}(e) \neq \emptyset \}$ and edges $D_X = \{ (e, e') \in E_X \times E_X | \exists_{o \in X, 1 \leq i < n} \langle e_1, ..., e_n \rangle = \pi_{trace}(o) \wedge e = e_i \wedge e' = e_{i+1} \}$.
The set $px(L) = \{ P_X | X \in con(L) \}$ contains all process executions of event log $L$.
\end{definition}

\autoref{fig:px} shows the process execution of the running example.
Object-centric Petri nets describe object-centric behavior by using types like a colored Petri net \cite{jensen_coloured_2009}.

$\mathcal{B}(A)$ is used to represent all multisets for a set $A$.
Given multiset $M$ for set $A$, the number of instances of element $a \in A$ in $M$ is $M(a)$.
We overload the notation $M = [a^k | a \in A]$ to state that there are $k$ instances of element $a$ in multiset $M$.

\begin{definition}[Object-centric Petri Net \cite{aalst_discovering_2020}]
    An object-centric Petri net is a tuple $ON = (N, pt, F_{var})$ where $N = (P, T, F, l)$ is a labeled Petri net with places P and transitions T.
    $F \in \mathcal{B}((P \times T) \cup (T \times P))$ is the multiset of arcs between places and transitions.
    Transitions are labeled with activities or $\tau$ by $l: T \to \mathbb{U}_{act} \cup \{\tau\}$ with invisible activity $\tau \not \in \mathbb{U}_{act}$.
    $pt: P \to \mathbb{U}_{typ}$ maps places to object types and $F_{var} \leq F$ is the sub-multiset of variable arcs.

    Note that we label all transitions to activities or $\tau$ with function $l$.
    Other common definitions for object-centric Petri nets define $l$ as a partial function.
    This can be translated into our definition by assuming $l(t) = \tau$ for all $t$ without a label.
    We define the following derived notations for object-centric Petri nets
    \begin{itemize}
        \item $\bullet t = \{ p \in P | (p, t) \in F\}$ is the preset of transition $t \in T$.
        \item $t \bullet = \{ p \in P | (t, p) \in F\}$ is the post set of transition $t \in T$.
        \item $pl(t) = \bullet t \cup t \bullet$ are the input and output places of $t \in T$, 
        $pl_{var}(t) = \{ p \in P | \{ (p,t), (t,p) \} \cap F_{var} \neq \emptyset \}$
        are places that are connected through variable arcs and
        $pl_{nv}(t) = \{ p \in P | \{ (p,t), (t,p) \} \cap (F \setminus F_{var}) \neq \emptyset \}$ are places that are connected through non-variable arcs.
        \item $tpl(t) = \{ pt(p) | p \in pl(t) \}$, $tpl_{var}(t) = \{ pt(p) | p \in pl_{var}(t) \}$, and
        $tpl_{nv}(t) = \{ pt(p) | p \in pl_{nv}(t) \}$ are object types related to transitions.
    \end{itemize}
\end{definition}

\autoref{fig:pn-w-flattened} shows the object-centric Petri net for the running example.
It has two object types and variable arcs for transitions \textit{receive sample order} and \textit{receive products order}.

\begin{definition}[Well-Formed Object-Centric Petri Net \cite{aalst_discovering_2020}]
    Let $ON = (N, pt, F_{var})$ be an object-centric Petri net with $N = (P, T, F, l)$.
    ON is well-formed if for each transition $t \in T: tpl_{var}(t) \cap tpl_{nv}(t) = \emptyset$. 
\end{definition}

In a well-formed object-centric Petri net arcs, connected to the same transition and places with the same object type, are either all variable or none of them is.
We assume for the following that all the object-centric Petri nets we use are well-formed.
Similar to colored Petri nets, object-centric Petri nets use the notion of markings and bindings to describe the semantics of a Petri net.

\begin{definition}[Marking of object-centric Petri Net \cite{aalst_discovering_2020}]
    Let $ON = (N, pt, \allowbreak F_{var})$ be an object-centric Petri net with $N = (P, T, F, l)$.
    $\mathcal{Q}_{ON} = \{ (p, o) \in P \times \mathbb{U}_{obj} | pt(p) = \pi_{type}(o) \}$ is the set of possible tokens.
    A marking M of ON is a multiset of tokens $M \in \mathcal{B}(\mathcal{Q}_{ON})$.
\end{definition}

A binding describes which transition fires and what object instances are consumed and produced per object type.

\begin{definition}[Binding of object-centric Petri Net \cite{aalst_discovering_2020}]
    Let $ON = (N, pt,\allowbreak F_{var})$ be an object-centric Petri net with $N = (P, T, F, l)$.
    The set of all possible bindings is 
    $B = \{ (t,b) \in T \times (\mathbb{U}_{type} \not \to \mathcal{P} (\mathbb{U}_{obj})) | dom(b) = tpl(t) \land \forall_{ot \in tpl_{nv}(t)} \forall_{p \in pl_{nv}(t), pt(p) = ot} |b(ot)| = F(p,t) \}$.
    A binding $(t, b) \in B$ corresponds to firing transition t in Petri net ON.
    The object map b describes what object instances are consumed and produced.
    The multiset of consumed tokens given binding $(t, b) \in B$ is
    $cons(t,b) = [(p, o) \in \mathcal{Q}_{ON} | p \in \bullet t \land o \in b(pt(p))]$.
    The multiset of produced tokens given binding $(t, b) \in B$ is
    $prod(t,b) = [(p, o) \in \mathcal{Q}_{ON} | p \in t \bullet \land o \in b(pt(p))]$.

    Binding $(t, b) \in B$ is enabled in marking $M \in \mathcal{B}(\mathcal{Q}_{ON})$ if $cons(t,b) \leq M$.
    Applying binding $(t,b)$ in marking $M$ leads to new marking $M' = M - cons(t,b) + prod(t,b)$.
    We use the notation $M \xrightarrow{\text{(t,b)}} M'$ for applying $(t,b)$ in $M$.
    This implies that $(t,b)$ was enabled in $M$ and $M'$ is the result of applying $(t,b)$ in $M$.

    This notation can be extended to a sequence of bindings
    $\sigma = \langle (t_1, b_1), (t_2, b_2), ..., \allowbreak (t_n, b_n) \rangle \in B^*$ such that
    $M_0 \xrightarrow{\text{$(t_1,b_1)$}} M_1 \xrightarrow{\text{$(t_2,b_2)$}} M_2 ... \xrightarrow{\text{$(t_n,b_n)$}} M_n$.
    We use the notation  $M \xrightarrow{\text{$\sigma$}} M'$ to show that $M'$ can be reached from $M$ by applying the bindings in $\sigma$ after another.
    The transitions can be mapped to activities using the label function $l$.
    This results in the visible binding sequence
    $\sigma_{\upsilon} = \langle (l(t_1),b_1), (l(t_2),b_2), ..., (l(t_n),b_n))$ where $(l(t_i),b_i)$ is omitted if $l(t_i) = \tau$.
\end{definition}

\begin{definition}[Accepting object-centric Petri Net \cite{aalst_discovering_2020}]
    An accepting object-centric Petri net is a tuple $AN = (ON, M_{init}, M_{final})$ where $ON = (N, pt, F_{var})$ is a well-formed object-centric Petri net.
    $M_{init} \in \mathcal{B}(\mathcal{Q}_{ON})$ and $M_{final} \in \mathcal{B}(\mathcal{Q}_{ON})$ indicate the initial and final markings of the net.
\end{definition}

Accepting object-centric Petri nets accept some binding sequences and some not.
The set of all binding sequences that are accepted form a language.

\begin{definition}[Language of an Accepting Petri Net \cite{aalst_discovering_2020}]
    The language $\phi(AN) = \{ \sigma_{\upsilon} | M_{init} \xrightarrow{\text{$\sigma$}} M_{final} \}$ of an accepting object-centric Petri net
    $AN = \allowbreak (ON,\allowbreak M_{init}, M_{final})$ contains all the visible binding sequences starting in $M_{init}$ and ending in $M_{final}$.
\end{definition}

\section{Alignment}
\label{sec:alignment}
Alignments show how process executions and the allowed behavior of a de-jure model relate to each other.
We use moves to represent whether something occurs in the process execution, the de-jure model, or both of them.
Note that we do not allow alignments to alter the set of objects. We assume them to be fixed.

\begin{definition}[Moves]
    Let $L=(E, O, OT, \pi_{act}, \pi_{obj}, \pi_{times}, \pi_{trace})$ be an object-centric event log and $P_X = (E_X, D_X) \in px(L)$ a process execution.
    Let $AN = (((P, T, F, l), pt, F_{var}), M_{init}, M_{final})$ be an accepting object-centric Petri net.
    The set of all moves is $moves(P_X, AN) \subseteq (\{\pi_{act}(e) | e \in E_X \} \cup \{\gg\}) \times \mathcal{P}(O) \times (T \cup \{ \gg\}) \times \mathcal{P}(O)$ with skip symbol $\gg \not \in \mathbb{U}_{act} \cup T$. A move $(a_{log}, o_{log}, t_{mod}, o_{mod}) \in moves(P_X, AN)$ is one of the following three types:\\
    Log move - for an $e\in E_X$: $a_{log} = \pi_{act}(e)$, $ o_{log} = \pi_{obj}(e)$, $t_{mod} = \gg$, and $o_{mod} = \emptyset$.\\
    Model move - for a $(t,b) \in \sigma \text{ with } \sigma_{\upsilon} \in \phi(AN)$: $t_{mod} = t$, $o_{mod} = \bigcup_{o \in range(b)} o$, $a_{log} = \gg$, and $o_{log} = \emptyset$.\\
    Synchronous move - for an $e\in E_X$ and a $(t,b) \in \sigma \text{ with } \sigma_{\upsilon} \in \phi(AN)$: $a_{log} = \pi_{act}(e) = l(t)$, $ t_{mod} = t$, and $o_{log} = o_{mod} = \pi_{obj}(e) = \bigcup_{o \in range(b)} o$.
\end{definition}

Each type of move is presented in \autoref{fig:allowed-moves}.
For synchronous moves, the activity and objects of the model and log part have to be exactly the same.
For log and model moves, only one part has an activity and objects while the other parts are skipped.
The skip symbol $\gg$ represents that nothing happened in that part.
The upper part is the log part $a_{log}$ and $o_{log}$.
The lower block is the model part that contains $t_{mod}$, the activity $l(t_{mod})$ it is labeled with, and $o_{mod}$.

\begin{figure}[h]
    \centering
    \includegraphics[width=0.7\textwidth]{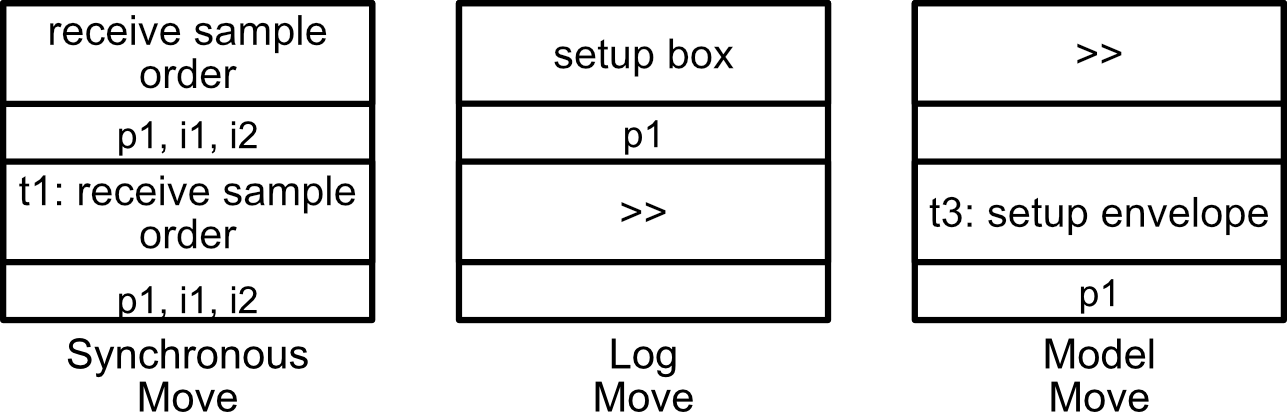}
    \caption{One synchronous move, one log move, and one model move for the running example with the accepting object-centric Petri net in \autoref{fig:pn-w-flattened} and the process execution in \autoref{fig:px}.}
    \label{fig:allowed-moves}
\end{figure}

We define the following projections on moves.

\begin{definition}[Move Projections]\\
    Given a move $m = (a_{log}, o_{log}, t_{mod}, o_{mod}) \in moves(P_X, AN)$ with process execution $P_X$ and accepting object-centric Petri Net $AN = (((P, T, F, l), pt, F_{var}),\allowbreak M_{init}, \allowbreak M_{final})$.
    We use the following projections to map moves to their attributes:\\
    $\pi_{la}(m) = a_{log}$ maps moves to their log activity.\\
    $\pi_{lo}(m) = o_{log}$ maps moves to their log objects.\\
    $\pi_{mt}(m) = t_{mod}$ maps moves to their model transition.\\
    $\pi_{ma}(m) = l(t_{mod})$ maps moves to the activity the transition is labeled with.\\
    $\pi_{mo}(m) = o_{mod}$ maps moves to their model objects.\\

\end{definition}

When reasoning about the model or log behavior individually, we want to ignore skipped behavior in that part.
An alignment, which we define in \autoref{def:alignment}, is a directed acyclic graph of moves.
We need to reason about the model and log behavior individually to define alignments.
Therefore, we introduce the following reductions that remove moves with skipped behavior in a given part from a directed acyclic graph of moves while maintaining the partial order defined by the acyclic graph.

\begin{definition}[Reduction to Log and Model Part]
    Let  $MG = (M, C)$ be a directed acyclic graph with vertices $M \subseteq moves(P_X, AN)$ and edges $C \subseteq M \times M$ with process execution $P_X$ and accepting object-centric Petri Net $AN$.
    The reduction to moves with visible activity in the log part is $MG_{\downarrow log} = (M_{\downarrow log}, C_{\downarrow log})$ with:
    \begin{itemize}
        \item $M_{\downarrow log} = \{ m \in M | \pi_{la}(m) \neq \gg\}$ synchronous and log moves.
        \item $C_{\downarrow log} = \{ (m_1, m_n) \in M_{\downarrow log} \times M_{\downarrow log} | \exists_{<m_1, ..., m_n> \in M^*}\; \forall_{1 \leq i < n}\; (m_i, m_{i+1}) \in C \wedge \forall_{1 < i < n}\;\allowbreak \pi_{la}(m_i) \allowbreak = \gg \}$ edges between synchronous and log moves and new edges where model moves were removed.
    \end{itemize}

    The reduction to moves with visible activity in the model part is $MG_{\downarrow mod} = (M_{\downarrow mod}, C_{\downarrow mod})$ with:
    \begin{itemize}
        \item $M_{\downarrow mod} = \{ m \in M | \pi_{ma}(m) \neq \gg\}$ synchronous and model moves
        \item $C_{\downarrow mod} = \{ (m_1, m_n) \in M_{\downarrow mod} \times M_{\downarrow mod} | \exists_{<m_1, ..., m_n> \in M^*}\; \forall_{1 \leq i < n}\; (m_i, m_{i+1}) \in C \wedge \forall_{1 < i < n}\;\allowbreak \pi_{ma}(m_i) \allowbreak = \gg \}$ edges between synchronous and model moves and new edges where log moves were removed.
    \end{itemize}
\end{definition}

In \autoref{fig:reductions} both $MG_{\downarrow log}$ and $MG_{\downarrow model}$ are visualized for an arbitrary directed acyclic graph of moves.
$MG_{\downarrow log}$ describes a directed acyclic graph after removing all model moves and related edges.
New edges are added when two movements used to be connected via removed model moves in the movement graph.
$MG_{\downarrow model}$ behaves simultaneously for the model part.

\begin{figure}[h]
    \centering
    \includegraphics[width=0.9\textwidth]{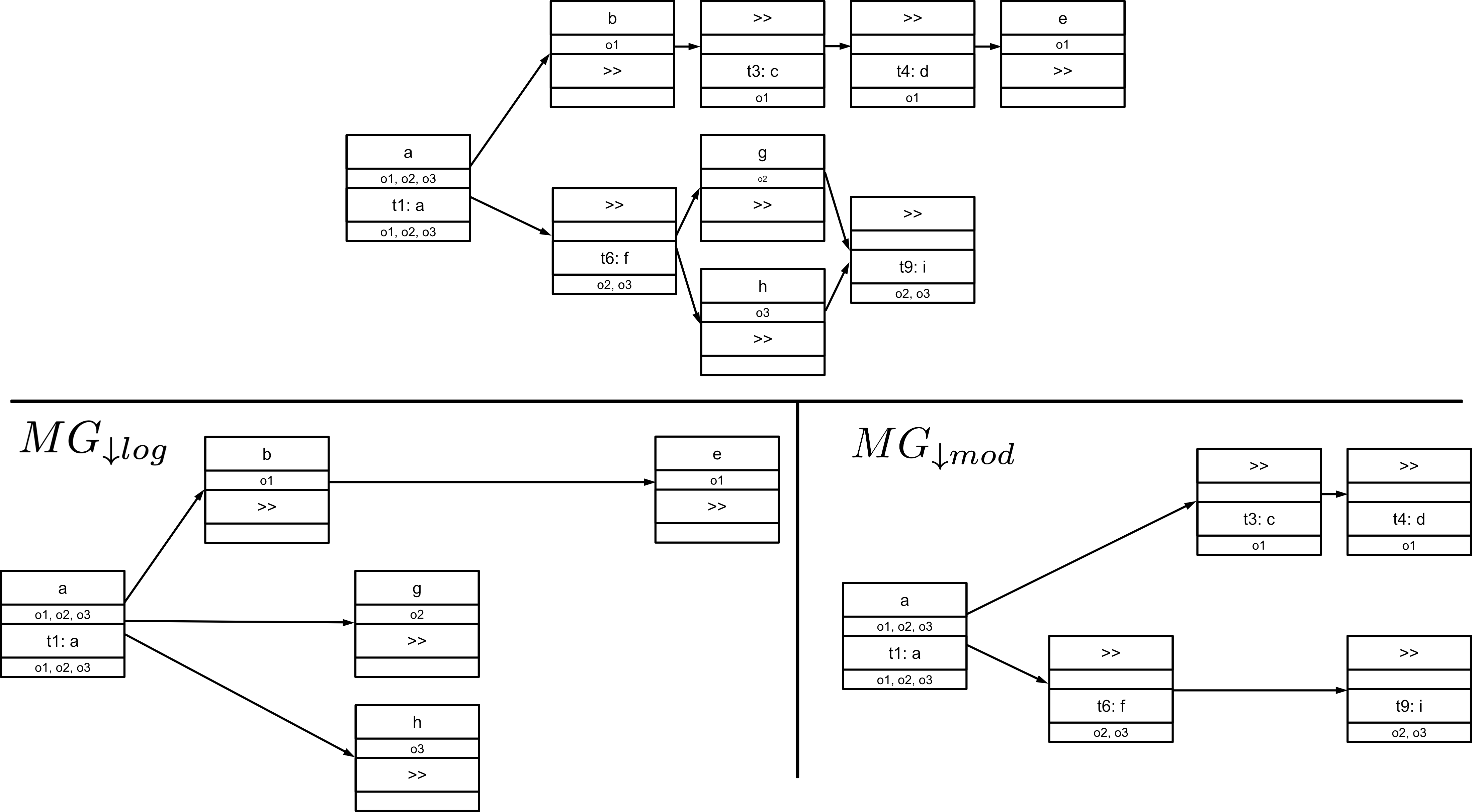}
    \caption{Reductions $MG_{\downarrow log}$ and $MG_{\downarrow model}$ for a directed acyclic graph of moves.}
    \label{fig:reductions}
\end{figure}

An alignment is a directed acyclic graph of moves that requires the log part to contain the process execution and the model part to be in the language of the de-jure model.

\begin{definition}[Alignment]
    \label{def:alignment}
    Let $L=(E, O, OT, \pi_{act}, \pi_{obj}, \pi_{times}, \pi_{trace})$ be an object-centric event log and $P_X = (E_X, D_X) \in px(L)$ a process execution.
    Let $AN = (((P, T, F, l), pt, F_{var}) M_{init}, M_{final})$ be an accepting object-centric Petri net.
    An alignment $AL_{P_X, AN} = (M, C)$ is a directed acyclic graph on $M \subseteq moves(P_X, AN)$ such that:

    The alignment contains the process execution behavior in the log parts:
    $P_X$  is isomorphic to $AL_{{P_{X}, AN}_{\downarrow log}}$ with bijective function $f: E_X \to M_{\downarrow log}$ such that $\forall_{e \in E_X} \pi_{act}(e) = \pi_{la}(f(e)) \land \pi_{obj}(e) = \pi_{lo}(f(e))$.

    The alignment contains behavior that is accepted by the Petri net in the model parts: 
    There exists a binding sequence $\sigma = \langle (t_1, b_1), (t_2, b_2), ..., (t_n, b_n) \rangle \in B^*$ with $\sigma_{\upsilon} \in \phi (AN)$ and a bijective function $f': B \to M_{\downarrow mod}$ such that:
        \begin{itemize}
            \item $\forall_{(t,b) \in \sigma} t = \pi_{mt}(f'(t,b)) \wedge \bigcup_{o \in range(b)} o = \pi_{mo}(f'(t,b))$
            \item $\forall_{m_1, m_2 \in M_{\downarrow mod}} (m_1, m_2) \in C_{\downarrow mod} \Rightarrow \exists_{1 \leq i < j \leq n} m_1 = f'(t_i, b_i) \land m_2 = f'(t_j, b_j)$
        \end{itemize}

    There can be multiple alignments for a process execution and an accepting object-centric Petri net.
    $al(P_X, AN)$ is the set of all these alignments.
\end{definition}

An alignment for the running example can be seen in \autoref{fig:alignment}.
The graph of moves is directed and acyclic which creates a partial order of moves.
When reducing the graph to log and synchronous moves, it is isomorphic to the given process execution in \autoref{fig:px}.
This ensures that the behavior of the given process execution is contained in the alignment.
The reduction to the model part relates to a binding sequence that has to be in the language of the accepting object-centric Petri net.
This ensures that the model part describes behavior that is accepted by the model.
Note that an alignment does not put any other requirement on the model part than to be allowed by the model.
Therefore, alignments can also contain model behavior that is very different from the process execution.
Those alignments will end up with more model and log moves and fewer synchronous moves.
Cost functions for moves allow us to search for special behavior from the model that we want to align.
Most of the time we are looking for the allowed behavior that is the most similar to the given process execution.
That is what the standard cost function formalizes.

\begin{definition}[Standard Cost of Move Function]\\
Let $AL_{P_X, AN} = (M, C) \in al(P_X, AN)$ be an alignment with process execution $P_X$ and accepting object-centric Petri Net $AN$.
The cost function $cost_{move}: moves(P_X, AN) \to \mathbb{R}$ is defined as:\\
$cost_{move}(m) =
\begin{cases}
0 \; \text{if m is a synchronous move,}\\ 
|\pi_{lo}(m) \cup \pi_{mo}(m)| \; \text{if m is a model or log move} \wedge \pi_{ma}(m) \neq \tau,\\
\varepsilon \; \text{if } \pi_{ma}(m) = \tau \land a_{log} = \gg ,\\
+\infty \; \text{else}
\end{cases}
$

With $\varepsilon$ being a positive very small number.
The cost of a complete alignment is the sum over all alignment moves: $cost_{alignment}(AL_{P_X, AN}) = \sum_{m \in M} cost_{move}(m)$.
    
\end{definition}

The cost of the alignment in \autoref{fig:alignment} is 6 because there are 3 model and 3 log moves that have one object each.
The lower the cost the better model and log part match, because synchronous moves are cheaper than model and log moves.
We call one of the cheapest alignments an optimal alignment.

\begin{definition}[Optimal Alignment]
    Let $L$ be an object-centric event log and $P_X \in px(L)$ be a process execution.
    Let $AN$ be an accepting object-centric Petri net.
    An alignment $AL_{P_X, AN} = (M, C) \in al(P_X, AN)$ is optimal if\\ $\forall_{a \in al(P_X, AN)} \allowbreak cost_{alignment}(AL_{P_X, AN}) \leq cost_{alignment}(a)$.
\end{definition}

Note that there can still be multiple optimal alignments for a given process execution and a de-jure Petri net as long as they are all equally similar to the process execution.

\begin{figure}[h]
    \centering
    \includegraphics[width=0.7\textwidth]{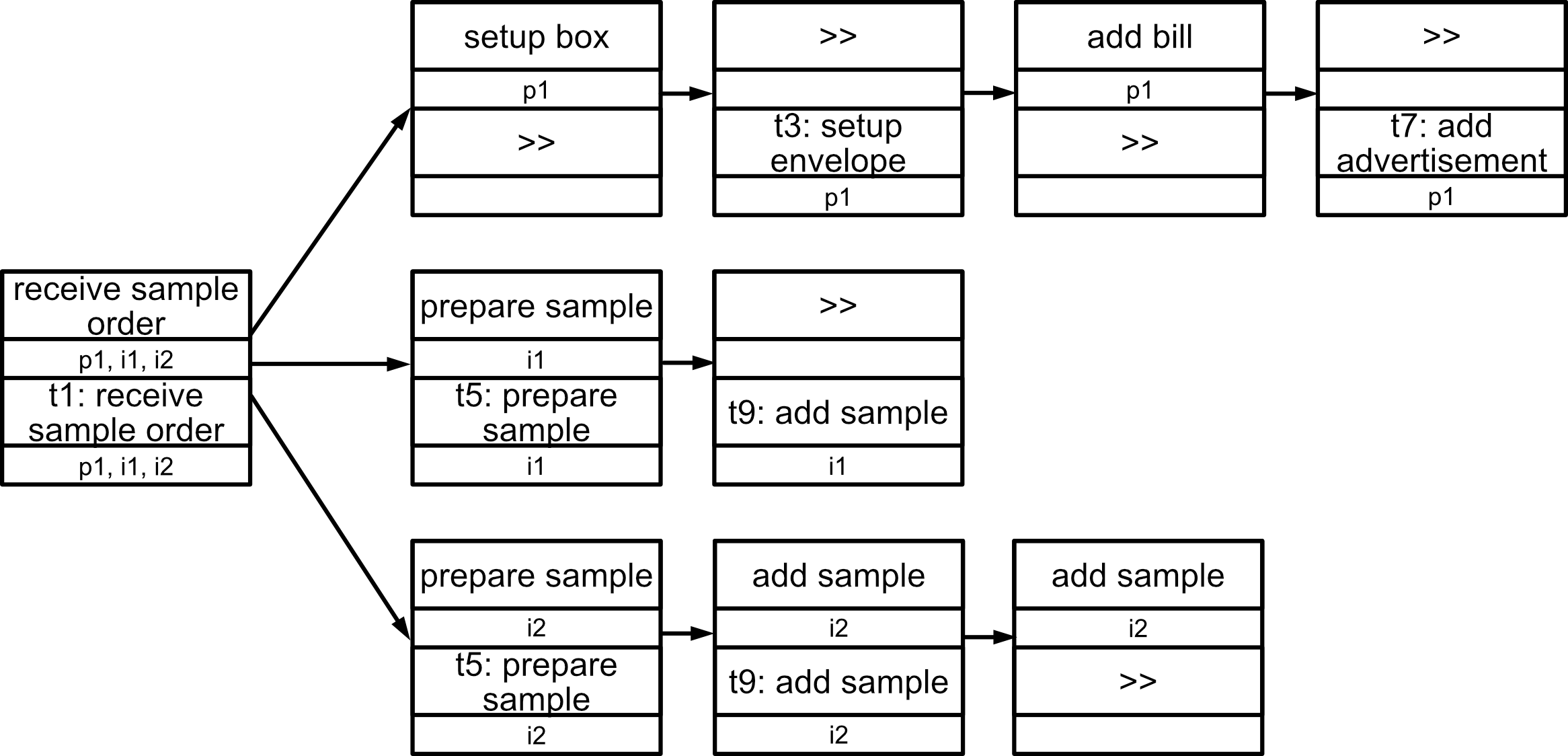}
    \caption{Alignment for process execution in \autoref{fig:px} and accepting object-centric Petri net in \autoref{fig:pn-w-flattened}.}
    \label{fig:alignment}
\end{figure}

An optimal object-centric alignment for the running example can be seen in \autoref{fig:alignment}.
The inter-object dependencies that are defined in the object-centric Petri net in \autoref{fig:pn-w-flattened} are respected by the object-centric alignment.
For example, the object-centric alignment agreed on one shared start activity for  \textit{p1}, \textit{i1}, and\textit{i2} which keeps the alignment consistent with inter-object dependencies.
This differentiates object-centric alignments and traditional alignments which can violate this requirement.

\section{Object-Centric Synchronous Product Net}
\label{sec:sny-prod-net}
The first part of our approach to calculate optimal alignments is to create a synchronous product net for a given process execution and an accepting object-centric Petri net.
That synchronous product net is designed to generate all possible alignments.
Since alignments can have model, log, and synchronous moves, the synchronous product net consists of three parts that directly relate to them.
In the synchronous product net in \autoref{fig:syn-pn} these parts are marked.
First, we construct the log part from the process execution.
Then, we pre-process the de-jure model to finally merge them together to the synchronous product net and add the synchronous part.

\subsection{Process Execution Net Construction}
The process execution net will be the part of the synchronous product net that guarantees that the process execution is contained in the alignment.
The construction of the process execution net relates to Petri net runs \cite{desel_placetransition_1998} and causal nets \cite{aalst_causal_2011}.
For each object, each edge in the process execution defines a precondition for an event.
The process execution net in \autoref{fig:px-net}, therefore, has a place for each condition defined in the process execution.

\begin{figure}[h]
    \centering
    \includegraphics[width=0.9\textwidth]{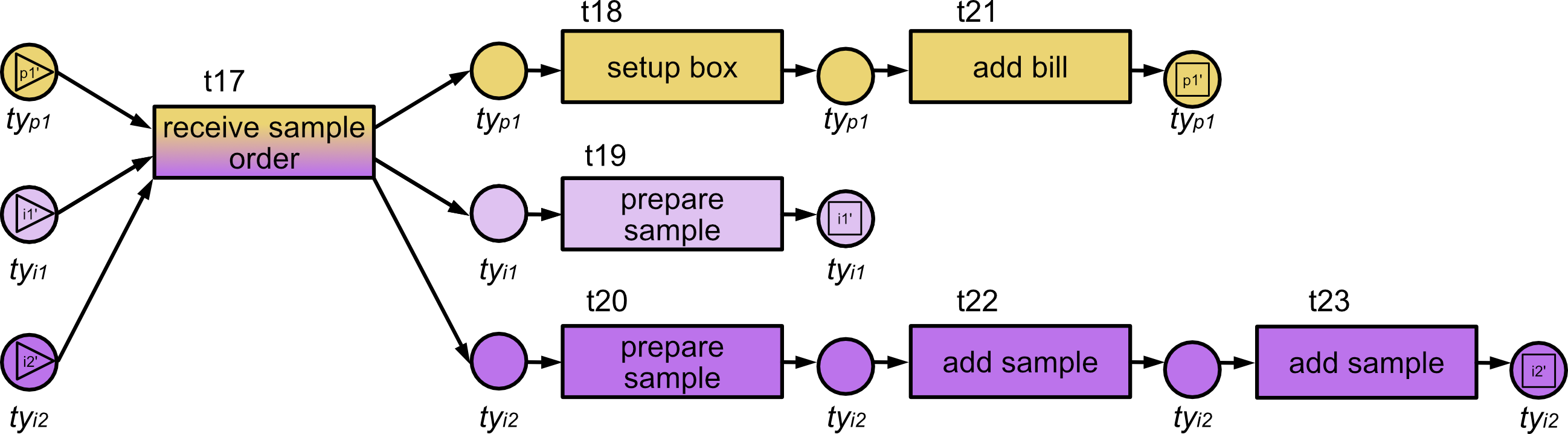}
    \caption{Process Execution Net}
    \label{fig:px-net}
\end{figure}

Because object-centric Petri nets are under-specified and we will merge the created net with a given accepting object-centric Petri net, we can not use the same objects and types for both, the log and model part.
The process execution contains \textit{directly follows} relations on the level of object instances.
For example, the running example process execution in \autoref{fig:px} shows that on item \textit{i1} \textit{add sample} never happens, whereas on item \textit{i2} activity \textit{add sample} happens twice.
Although they both have the same object type, the process execution differentiates them.
An accepting object-centric Petri net is under-specified in that regard because it only differentiates places by type.
Clearly, it is important to separate the object instances because otherwise, deviations in one object instance could compensate for a contrary deviation in another.
For example, the missing \textit{add sample} activity on \textit{i1} could be compensated by the additional \textit{add sample} activity on \textit{i2} if one would not strictly separate which event happened on which object instance.
To enable an accepting object-centric Petri net to separate object instances, we have to create a new individual type for each object instance in the process execution.
Since each object instance must have exactly one type, we then need to replace the original object instances with new individual placeholder object instances in the process execution net.

\begin{definition}[Generation of new Objects and Types]
    Let $L=(E, O, OT, \allowbreak \pi_{act}, \pi_{obj}, \pi_{times}, \pi_{trace})$ be an object-centric event log and $P_X = (E_X, D_X) \in px(L)$ a process execution containing objects $X$.
    Let $O_{new}(P_X) \in \mathbb{U}_{obj}$ with $O_{new}(P_X) \cap O = \emptyset$ and $|O_{new}(P_X)| = |X|$ be a set of new object instances 
    and $OT_{new}(P_X) \in \mathbb{U}_{type}$ with $OT_{new}(P_X) \cap \{ \pi_{type}(o) | o \in X \} = \emptyset$ and $|OT_{new}(P_X)| = |X|$ be a set of new types.
    
    Let $new_{obj}: X \to O_{new}(P_X)$ be a bijective function that renames objects from $X$ to unused objects in $O_{new}(P_X)$.
    Let $new_{type}: X \to OT_{new}(P_X)$ be a bijective function that relates each object to its unique new type.

    We define the following derivative concepts:\\
    $orob: O_{new}(P_X) \to X$ with $orob = new_{obj}^{-1}$ returns the original object.\\
    $orty: OT_{new}(P_X) \to \{ \pi_{type}(o) | o \in X \}$ with $orty(ot)  = \pi_{type}(new_{type}^{-1}(ot))$ returns the original type of the object that is associated with the given new type.
\end{definition}

We create the process execution net with the new objects and types.
The conditions defined in the process execution are the \textit{directly follows} relation per object.
For each condition, the Petri net contains a place.
Also, start and end places are added for each object.
The transitions relate to the events of the process execution.

\begin{definition}[Process Execution Net]
    Let $L = (E, O, OT, \pi_{act}, \pi_{obj}, \pi_{time},\allowbreak \pi_{trace})$ be an object-centric event log.
    Let $P_X = (E_X, D_X) \in px(L)$ be a process execution from the event log.
    The process execution net $PX_{net} = (((P, T, F, l), pt,\allowbreak F_{var}), M_{init},\allowbreak M_{final})$ is an accepting object-centric Petri net with:

    \begin{itemize}
        \item $P = \{ p_{o;i} | o \in X \wedge \pi_{trace}(o) = \langle e_1, ..., e_n \rangle \wedge 1 \leq i \leq n-1 \}\\ 
        \cup \{ p_{o;s} | o \in X \} \cup \{ p_{o;e} | o \in X \}$
        \item $T = \{ t_e | e \in E_X \}$
        \item $F = [(t_e, p_{o;i}) \in T \times P | \pi_{trace}(o) = \langle e_1, ..., e_n \rangle \wedge \exists_{1 \leq i \leq n-1} e_i = e ]\\
        \cup [ (p_{o;i}, t_e) \in P \times T | \pi_{trace}(o) = \langle e_1, ..., e_n \rangle \wedge \exists_{1 \leq i \leq n-1} e_{i+1} = e ]\\
        \cup [ (p_{o;s}, t_e) \in P \times T | \pi_{trace}(o) = \langle e_1, ..., e_n \rangle \wedge e_1 = e ]\\
        \cup [ (t_e, p_{o;e}) \in T \times P | \pi_{trace}(o) = \langle e_1, ..., e_n \rangle \wedge e_n = e ]$
        \item $l(t_e) = \pi_{act}(e)$; $pt(p_{o;i}) = new_{type}(o)$; $F_{var} = \emptyset$
        \item $M_{init} = \{(p_{o;s}, new_{obj}(o)) | o \in X \}$; $M_{final} = \{ (p_{o;e}, o') | new_{obj}(o) \in X \}$
    \end{itemize}
\end{definition}

\subsection{Pre-processing of the Object-Centric Petri Net}

The de-jure behavior is already given as an accepting object-centric Petri net, which is very close to what we need to create the synchronous product net.
But an accepting object-centric Petri net can have variable arcs.
The Petri net of the running example in \autoref{fig:pn-w-flattened} has variable arcs for $t_1$ and $t_2$.
For transitions with variable arcs, we do not know beforehand how many objects they will use.
This becomes a problem when we try to find synchronous moves between the process execution net and the de-jure net.
A transition can only be synchronous if they use the same object instances which implies that they use the same number of object instances.
Therefore we need to know beforehand how many object instances a transition will use.

As mentioned in \autoref{sec:alignment}, we assume that the set of objects for the alignment is immutable and defined by the process execution.
For a predefined set of object instances, the number of objects a variable arc can use is finite.
Therefore, we can replace transitions with variable arcs with a set of transitions without variable arcs.
For each combination of how many objects a variable arc could consume we can add a new transition to the Petri net that uses exactly that number of objects, but by replacing the variable arc with a number of non-variable arcs.
When doing this for the Petri net in \autoref{fig:pn-w-flattened} the result will be the Petri net without variable arcs in \autoref{fig:pre-pn}.

\begin{definition}[Pre-processing of Accepting Object Centric Petri Net]
    Let $L=(E, O, OT, \pi_{act}, \pi_{obj}, \pi_{times}, \pi_{trace})$ be an object-centric event log and $P_X = (E_X, D_X) \in px(L)$ a process execution.
    Let $AN = (((P, T, F, l), pt, F_{var}),\allowbreak M_{init}, M_{final})$ be an accepting object-centric Petri net.
    Let there be an arbitrary but fixed order of types $ot_1, ..., ot_n \in OT$.
    The pre-processed accepting object-centric Petri net is $DJ_{net} = (((P', T', F', l'), pt', F'_{var}) M'_{init}, M'_{final})$ with:

    \begin{itemize}
        \item $P' = P$, $pt' = pt, F'_{var} = \emptyset$,
        \item $T' = \{t_{c_1, ..., c_n} | t \in T \wedge \forall_{1\leq i \leq n} 0 \leq c_{i} \leq |\{ o \in X | \pi_{type}(o)=ot_i \}| \text{ if }ot_i \in tpl_{var}(t) \text{ otherwise }
        c_i = 0\}$
        \item $F' = [(p,t_{c_1,...,c_n}) \in P' \times T' | (p,t) \in F \wedge (p,t) \not \in F_{var})]$\\
        $\cup [(t_{c_1,...,c_n}, p) \in T' \times P' | (t,p) \in F \wedge (t,p) \not \in F_{var})]$\\
        $\cup [(p,t_{c_1,...,c_n})^k \in P' \times T' | (p,t) \in F_{var} \wedge \exists_{1 \leq i \leq n} pt(p) = ot_i \wedge c_i = k]$\\
        $\cup[(t_{c_1,...,c_n}, p)^k \in T' \times P' | (t,p) \in F_{var} \wedge \exists_{1 \leq i \leq n} pt(p) = ot_i \wedge c_i = k]$
        \item $l'(t_{c_1,...,c_n}) = l(t)$
        \item $M'_{init} = [ (p,o) \in P' \times X| \forall_{t \in T} (t,p) \not \in F  \wedge pt(p) = \pi_{type}(o)]$
        \item $M'_{final} = [ (p,o) \in P' \times X| \forall_{t \in T} (p,t) \not \in F  \wedge pt(p) = \pi_{type}(o)]$
    \end{itemize}
    
\end{definition}

We call the pre-processed accepting object-centric Petri net a de-jure net.
The de-jure net of the running example can be seen in \autoref{fig:pre-pn}.

\begin{figure}[h]
    \centering
    \includegraphics[width=0.7\textwidth]{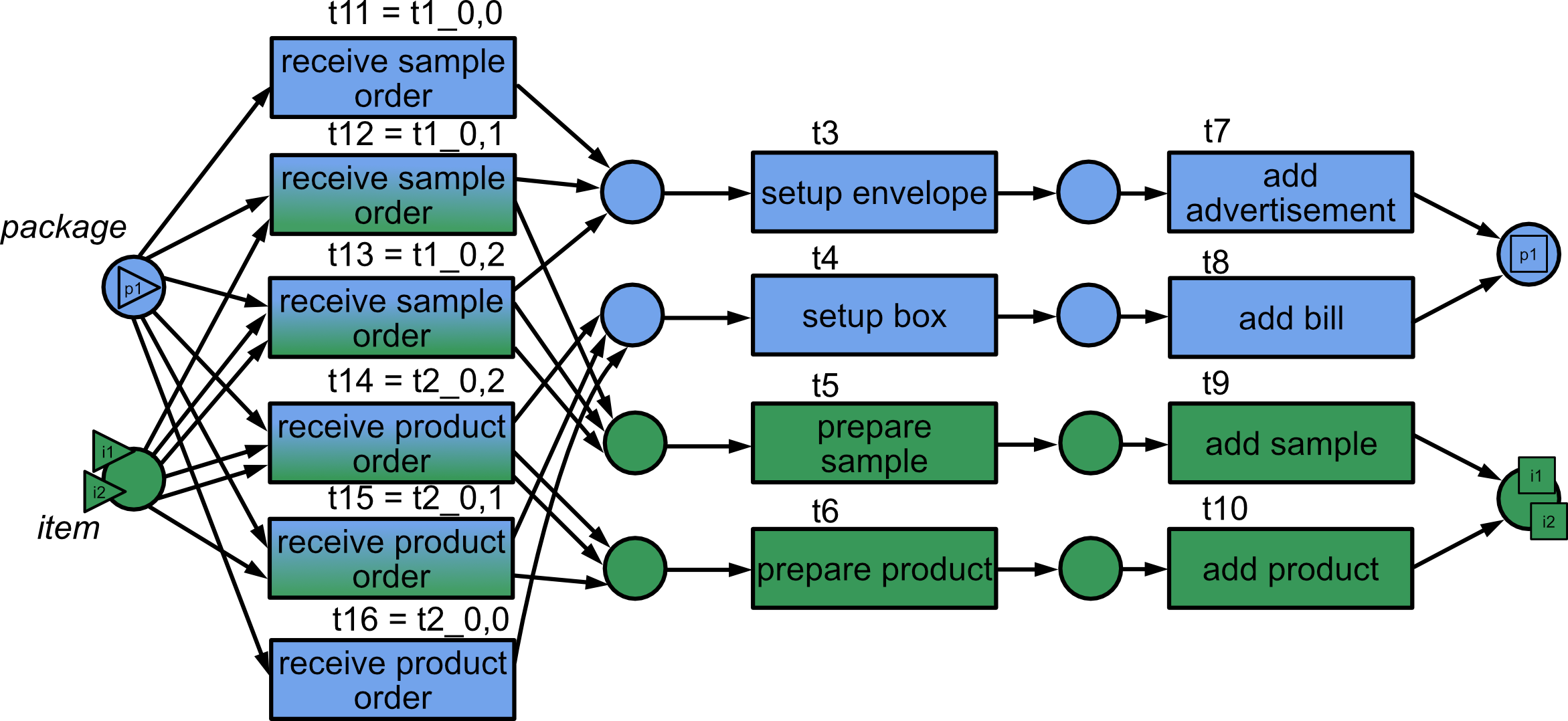}
    \caption{Preprocessed accepting object-centric Petri net without variable arcs}
    \label{fig:pre-pn}
\end{figure}

\subsection{Creating the Synchronous Product Net}

Now that we have two accepting object-centric Petri nets without variable arcs, one describing the process execution, and one describing the de-jure behavior, we can add them together to create the synchronous product net.
The synchronous product net should directly relate to possible alignments.
In an alignment, either the process execution and the de-jure parts advance forward separately or they advance forward synchronously.
For activities to happen in the process execution and the de-jure part simultaneously, they have to be the same activity on the same object instances.
Object-centric Petri nets can only separate between object types and not between object instances.
Therefore, we extend the object-centric Petri net with $\nu$-net requirements that can differentiate between object instances.
We decode log, model, and synchronous transitions with $(t_{log}, \gg), (\gg, t_{mod}),$ and $(t_{log}, t_{mod})$ as one can see in the synchronous product net in \autoref{fig:syn-pn}.

\begin{figure}[h]
    \centering
    \includegraphics[width=0.85\textwidth]{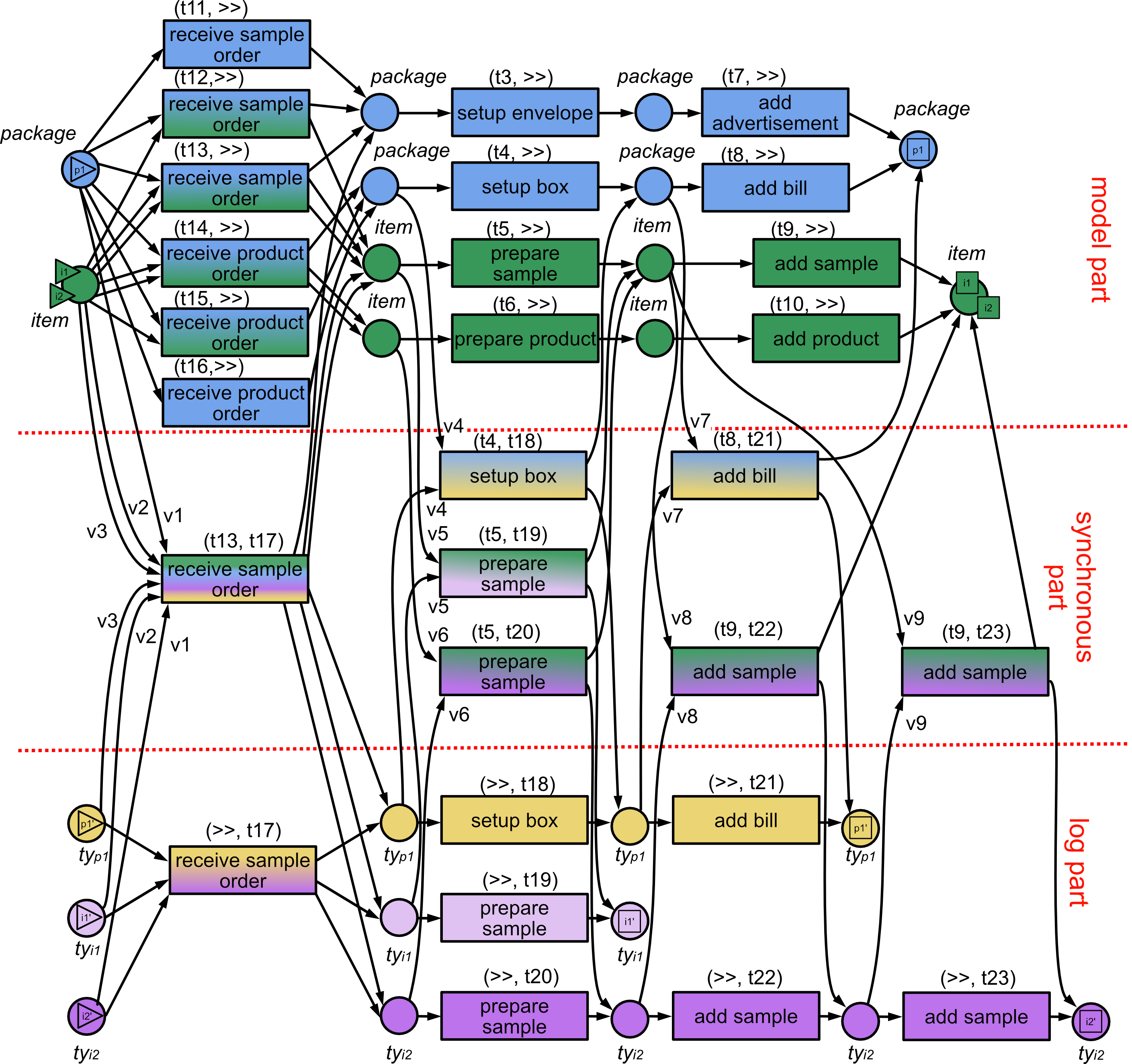}
    \caption{Synchronous product net for the process execution net in \autoref{fig:px-net} and accepting object-centric Petri net in \autoref{fig:pre-pn}}
    \label{fig:syn-pn}
\end{figure}

\begin{definition}[Synchronous Product Net]\\
    Let $PX_{net} = (((P_{PX}, T_{PX}, F_{PX}, l_{PX}), pt_{PX}, F_{var_{PX}} = \emptyset), M_{init_{PX}}, M_{final_{PX}})$ be a process execution net.
    Let $DJ_{net} = (((P_{DJ}, T_{DJ}, F_{DJ}, l_{DJ}), pt_{DJ}, F_{var_{DJ}} = \emptyset), M_{init_{DJ}}, M_{final_{DJ}}$ be a de-jure net.
    Let $Var$ be a set of unused variable names.

    Two transitions can be synchronous if they are labeled with the same activity and consume the same multiset of tokens. This requires them to be related to the same types and consume and produce the same number of tokens.\\
    $T_{syn} = \{(t_{PX}, t_{DJ}) \in T_{PX} \times T_{DJ} | 
    l_{PX}(t_{PX}) = l(_{DJ}(t_{DJ}) \wedge 
    \{orty(ty) | ty \in tpl(t_{PX}) \} = tpl(t_{DJ}) \wedge 
    \forall_{p_{PX} \in \bullet t_{PX}, p_{DJ} \in \bullet t_{DJ}} (orty(pt_{PX}(p_{PX})) = pt_{DJ}(p_{DJ})) \Rightarrow F_{PX}((p_{PX}, t_{PX})) {=} F_{DJ}((p_{DJ}, t_{DJ}))) \wedge 
    \forall_{p_{PX} \in t_{PX} \bullet, p_{DJ} \in t_{DJ} \bullet} (orty(pt_{PX}(p_{PX})) \allowbreak = pt_{DJ}(p_{DJ})) \Rightarrow 
    (F_{PX}((t_{PX}, p_{PX})) {=} F_{DJ}((t_{DJ}, p_{DJ})))
    \}$ is the set of synchronous transitions.
    
    The synchronous product net is the accepting object-centric $\nu$-Petri net $SP_{net} = (ON_{SP}, M_{init_{SP}}, \allowbreak M_{final_{SP}}, \allowbreak \nu)$ with $ON_{SP} = (N_{SP}, pt_{SP}, F_{var_{SP}} = \emptyset)$ with $N_{SP} = (P_{SP}, T_{SP}, F_{SP}, l_{SP})$ such that:

    \begin{itemize}
        \item $P_{SP} = P_{PX} \cup P_{DJ}$ is the union of places and $pt_{SP} = pt_{PX} \cup pt_{DJ}$
        \item $T_{SP} \subset (T_{PX}^{\gg} \times T_{DJ}^{\gg})= \{ (t_{PX}, \gg) | t_{PX} \in T_{PX}\} \cup \{ (\gg, t_{DJ}) | t_{DJ} \in T_{DJ}\} \cup T_{syn}$ is the union of transitions with additional synchronous transitions.
        \item $F_{SP} = [ (p_{PX}, (t_{PX}, \gg)) | (p_{PX}, t_{PX}) \in F_{PX} ]
        \cup [ (p_{DJ}, (\gg, t_{DJ})) | (p_{DJ}, t_{DJ}) \in F_{DJ} ]
        \cup [ ((t_{PX}, \gg), p_{PX}) | (t_{PX}, p_{PX}) \in F_{PX} ]
        \cup [ ((\gg, t_{DJ}), p_{DJ}) | (t_{DJ}, p_{DJ}) \in F_{DJ} ] 
        \cup [ (p_{PX}, (t_{PX}, t_{DJ})) | (p_{PX}, t_{PX}) \in F_{PX} \wedge (t_{PX}, t_{DJ}) \in T_{SP} ]\\
        \cup [ (p_{DJ}, (t_{PX}, t_{DJ})) | (p_{DJ}, t_{DJ}) \in F_{DJ} \wedge (t_{PX}, t_{DJ}) \in T_{SP} ] \\
        \cup [ ((t_{PX}, t_{DJ}), p_{PX}) | (t_{PX}, p_{PX}) \in F_{PX} \wedge (t_{PX}, t_{DJ}) \in T_{SP} ] \\
        \cup [ ((t_{PX}, t_{DJ}), p_{DJ}) | (t_{DJ}, p_{DJ}) \in F_{DJ} \wedge (t_{PX}, t_{DJ}) \in T_{SP} ] $
        \item $l_{SP} = \{ ((t_{PX}, \gg), a) | (t_{PX}, a) \in l_{PX} \}
        \cup \{ (\gg, t_{DJ}), a) | (t_{DJ}, a) \in l_{DJ} \}
        \cup \{ ((t_{PX}, t_{DJ}), a) | (t_{PX}, a) \in l_{PX},  (t_{DJ}, a) \in l_{DJ} \}$
        \item $M_{init_{SP}} = M_{init_{PX}} \cup M_{init_{DJ}}$ and $M_{final_{SP}} = M_{final_{PX}} \cup M_{final_{DJ}}$
        \item $\nu: P_{SP} \times T_{SP} \not \to Var$ such that:\\
        $\forall_{(t_{PX}, t_{DJ}) \in T_{PX} \times T_{DJ}} \forall_{p_{PX} \in \bullet(t_{PX}, t_{DJ})} \exists_{(p_{DJ}, (t_{PX}, t_{DJ})) \in F_{SP}}\\
        (p_{DJ} \in P_{DJ} \wedge  pt(p_{DJ}) = orty(pt(p_{PX})))
        \Rightarrow\\
        (\nu ((p_{DJ}, (t_PX, t_{DJ}))) = \nu ((p_{PX}, (t_{PX}, t_{DJ}))))$
    \end{itemize}
\end{definition}

In the process execution part, we use the $orob$ function when comparing objects from the model and log part.
The $\nu$ net requirement function assigns variables to the in-going arcs of the synchronous transition.
For each in-going arc from the process execution net, there has to be one arc from every place that has the same original type, so that those arcs have the same unique variable assigned by $Var$.
This requires the synchronous move to use the same object instance in the process execution net and the de-jure net.
The variables ensure that transitions can only consume the same object instance for arcs with the same variable.
This refers to the original objects not the newly created ones in the process execution net.

\begin{definition}[Valid Binding in Synchronous Product Net]
    Let $SP_{net} = (((P_{SP}, T_{SP}, F_{SP}, l_{SP}), pt_{SP}, F_{var_{SP}}), M_{init_{SP}}, M_{final_{SP}}, \nu)$ be the synchronous product net.
    Let $SP_{net\setminus \nu}$ be the synchronous product net without $\nu$.
    A binding sequence $\sigma = \langle (t_1, b_2), ..., (t_n, b_n) \rangle \in B^*$ is valid for $SP_{net}$ if it is valid for $SP_{net \setminus \nu}$ and for every $(t_i, b_i) \in \sigma$ with $t_i \in T_{sny}$ that is connected to an arc with a variable assigned by $\nu$ it holds that arcs with the same variable consume the same object:\\
    $\forall_{p_{DJ}\in P_{DJ}, p_{PX} \in P_{PX}, \{p_{DJ}, p_{PX}\} \subseteq \bullet(t_i)}
    (\nu((p_{PX}, t_i)) = \nu((p_{DJ}, t_i))
    \Rightarrow
    \{orob(o) | o \in b_i(pt((p_{PX}, t_i))) \}
    \cap b(pt((p_{DJ}, t_i))) \neq \emptyset)$.
    
\end{definition}

The resulting synchronous Petri net for the running example can be seen in \autoref{fig:syn-pn}.

\section{Alignments from Synchronous Product Net}
\label{sec:alignm-from-syn-prod-net}
This section describes the second part of our approach to finding object-centric alignments.
As an input for this part, we have a synchronous product net in which every transition relates to a move of an alignment.
All binding sequences from the initial marking to the final marking in the synchronous product net relate to an alignment.
In this section, we now want to find an optimal alignment given a cost function for moves.
Every binding in the synchronous product net directly relates to a move.
A binding defines the transition and the used objects.
As described in \autoref{sec:sny-prod-net} the transitions already describe the type of move with $(\gg, t)$ being a model move, $(t, \gg)$ being a log move, and $(t_{PX}, t_{DJ})$ being a synchronous move.
The activity is also defined by the transition from the binding.
The objects used in the binding are the objects of the move.

Therefore, searching for an optimal alignment relates to searching for an optimal binding sequence from the initial to the final marking of the synchronous product net.
Interpreting markings as nodes and bindings as edges between markings, we can set up the search space as a graph.
The cost of the edges is then the given cost function applied to the move that relates to the binding.
This is a well-defined search problem we can solve with standard search algorithms for finding the cheapest or shortest path in a graph.

\begin{definition}[Weighted State Space Graph]\\
    Let $SP_{net} = (((P_{SP}, T_{SP}, F_{SP}, l_{SP}), pt_{SP}, F_{var_{SP}}), M_{init_{SP}}, M_{final_{SP}}, \nu)$ be a synchronous product net.
    Let $B_{SP}$ be all possible bindings in $SP_{net}$.
    The weighted search space graph for that Petri net is $G_{wg} = (V_{wg}, E_{wg}, w_{wg})$ with $V_{wg} \subseteq \mathcal{B}(\mathcal{Q}_{ON_{SP}})$, $E_{wg} \subseteq \mathcal{B}(\mathcal{Q}_{ON_{SP}}) \times \mathcal{B}(\mathcal{Q}_{ON_{SP}}) $ and $w_{wg}: E_{wg} \to \mathbb{R}$ such that:

    \begin{itemize}
        \item $V_{wg} = \{ M \in \mathcal{B}(\mathcal{Q}_{ON_{SP}})| \exists_{\sigma \in B_{SP}^*} M_{init_{SP}} \xrightarrow{\text{$\sigma$}} M \}$ is the set of all reachable markings in the synchronous product net.
        \item $E_{wg} = \{ (M, M')\in \mathcal{B}(\mathcal{Q}_{ON_{SP}}) \times \mathcal{B}(\mathcal{Q}_{ON_{SP}}) | \exists_{(t, b) \ in B_{SP}} M \xrightarrow{\text{$(t,b)$}} M'\}$ connects each marking to its directly reachable markings.
        \item $w_{wg}(M, M') = cost_{move}(a_{log}, o_{log}, t_{model}, o_{model})$ where $(t,b) \in B_{SP}$ with $M \xrightarrow{\text{$(t,b)$}} M'$ and $t = (t_{PX}, t_{DJ}) \in T_{SP}$ so that:
        \begin{itemize}
            \item $a_{log} = l_{SP}(t_{PX})$; $t_{model} = t_{DJ}$
            \item $o_{log} = range(b)$ if $a_{log} \neq \gg$ else $o_{log} = \emptyset$
            \item $o_{model} = range(b)$ if $a_{model} \neq \gg$ else $o_{model} = \emptyset$
        \end{itemize}
    \end{itemize}
\end{definition}

\begin{theorem}
\label{theo:finite}
Let $P_X$ be a process execution with a finite number of objects and events and a let $DJ_{net} = (((P_{DJ}, T_{DJ}, F_{DJ}, l_{DJ}), pt_{DJ}, F_{var_{DJ}} = \emptyset), M_{init_{DJ}}, \allowbreak M_{final_{DJ}})$ be a de-jure net with a finite number of reachable markings from the initial marking. The synchronous product net $SP_{net}$ for $P_X$ and $DJ_{net}$ has a finite number of reachable markings.
\end{theorem}

\begin{proof}[\hypertarget{proof:finite-search-space}{Finite Search Space}]
Let $P_X = (E_X, D_X)$ be a process execution with a finite number of  objects $o=|X|$ and events $e=|E_X|$. Let $DJ_{net} = (((P_{DJ}, T_{DJ}, \allowbreak F_{DJ}, l_{DJ}), pt_{DJ}, \allowbreak F_{var_{DJ}} = \emptyset), M_{init_{DJ}}, M_{final_{DJ}})$ be a de-jure model for the process execution with a finite number of reachable markings $m_{DJ} = |\{ M | \exists_{\sigma \in B_{DJ}^*} \allowbreak M_{init_{DJ}(PX)} \xrightarrow{\text{$\sigma$}} M\}|$ with $B_{DJ}$ being the set of possible bindings in $DJ_{net}$.

Step 1 of our method creates the process execution net $PX_{net} = (((P_{PX}, T_{PX}, \allowbreak F_{PX}, l_{PX}), \allowbreak pt_{PX},\allowbreak F_{var_{PX}}), M_{init_{PX}},\allowbreak M_{final_{PX}})$ from the process execution.
Let $B_{PX}$ be the set of all bindings in $PX_{net}$.
At most one place is added per event for each object together with one additional start place for each object.
Since the number of events and objects is finite the number of places in $PX_{net}$ is finite.
Each transition produces as many tokens as it consumes because it is connected to one output place for each input place.
So the number of tokens stays consistent all the time.
Therefore, the number of reachable markings $m_{PX} = |\{ M | \exists_{\sigma \in B_{PX}^*} M_{init_{PX}} \xrightarrow{\text{$\sigma$}} M\}| \leq (e+1)^{o}$ in $PX_{net}$ is finite.

Step 2 of our method preprocesses $DJ_{net}$.
Thereby, no places are added.
Transitions with variable arcs are replaced by a set of transitions that can only reach markings the original transition was able to reach as well.
Thus, the number of reachable markings stays the same and is therefore finite.

In Step 3, $DJ_{net}$ and $PX_{net}$ are merged together to form the synchronous product net.
No places are added and the synchronous transitions that are added, do not add new reachable markings because there always exists one transition from $DJ_{net}$ and one from $PX_{net}$ that together have the same effect on the markings of the net.
Therefore, the number of markings in the synchronous product net $SP_{net}$ for $P_X$ and the preprocessed $DJ_{net}$ is the product of the two nets: $m_{SP} = |\{ M | \exists_{\sigma \in B_{SP}^*} M_{init_{SP}} \xrightarrow{\text{$\sigma$}} M\}| = m_{DJ} * m_{PX}$ with $B_{SP}$ being the set of all possible bindings in $SP_{net}$.
Since $m_{DJ}$ and $m_{PX}$ are finite, the synchronous product net has a finite number of reachable markings.
\end{proof}

As stated in \autoref{theo:finite} and shown in \hyperlink{proof:finite-search-space}{the finite search space proof}, the weighted search space graph is finite given that the event log and the de-jure model have a finite size.
For a finite weighted graph, it is decidable whether a cheapest path exists, and if one exists, algorithms like the Dijkstra algorithm \cite{dechter_generalized_1985} can find them.
The only case in which there is no cheapest path is when there is no path at all. 
This means that the synchronous product net cannot reach the final state from the initial state. 
This can only occur if the de-jure model has no option to complete for the set of objects of the process execution. 
In other words, the de-jure model does not contain any allowed behavior related to the process execution. 
In this case, there can be no alignment and the user is informed that the de-jure model does not match the process execution.
This is similar to traditional alignments with a de-jure model that has no option to complete.
In the normal case where the de-jure model describes behavior related to the given process execution, an optimal alignment is determined with the help of the shortest path.

The resulting shortest path is a binding sequence from the initial to the final marking.
It relates directly to an object-centric alignment for the given process execution and the de-jure model.
Using the synchronous product net in \autoref{fig:syn-pn} the search for an optimal alignment results in the alignment in \autoref{fig:alignment}.
Thereby respecting the inter-object dependencies between \textit{package} and \textit{item}.
Also, object instances are correctly separated by the $\nu$-net requirements, so that deviations of different object instances can not compensate another.
That can be seen in \autoref{fig:alignment} where the missing \textit{add sample} activity for \textit{i1} and the additional \textit{add sample} activity for \textit{i2} are both identified as deviations.

There can be multiple binding sequences with the same cost.
It is guaranteed that one of the cheapest is found, but if there are multiple binding sequences with the same cost, it depends on the implementation which one is found.


\section{Evaluation}
\label{sec:evaluation}
We conducted an evaluation with real-world data from the BPI2017 challenge \cite{van_dongen_bpi_2017}.
The evaluation is split into a qualitative part and a quantitative evaluation of the run time.

\subsection{Qualitative}

The purpose of the qualitative evaluation is to evaluate whether the proposed approach does indeed give better insights into complex processes than existing approaches.
In the following, we will compare the presented object-centric alignment approach to traditional alignments \cite{adriansyah_aligning_2014} on a flattened event log \cite{aalst_discovering_2020}.

We used data from BPI2017 \cite{van_dongen_bpi_2017} for our real-world evaluation.
It is one of the few publicly available event logs that can be turned into an object-centric event log.
The two object types of that process are \textit{application} and \textit{offer}.
All variants have exactly one application and at least one offer.
An application can be canceled or an offer is accepted.

For this evaluation, we selected only the 4 most dominant variants that made up 19.6\% of process executions.
For simplicity and clarity, we removed the activities \textit{Submit}, \textit{Complete}, and \textit{Accept} because they only appear in the beginning before the processes for \textit{application} and \textit{offer} interact.
We discovered the de-jure Petri net from those 4 variants using the python library ocpa \cite{adams_ocpa_2022} which implements the discovery approach from van der Aalst and Berti \cite{aalst_discovering_2020}.
The implementation of our approach also uses ocpa \cite{adams_ocpa_2022}.
In the next step, we introduced noise in the log by removing and replacing events in the given process execution.

\begin{figure}[t]
    \centering
    \includegraphics[width=0.9\textwidth]{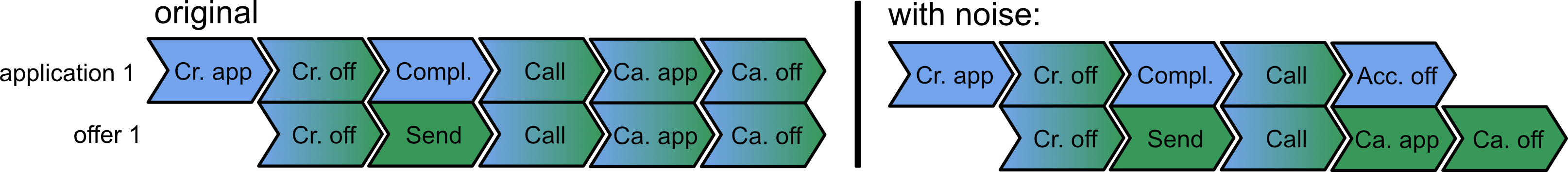}
    \caption{A process execution from \cite{van_dongen_bpi_2017} and a variation of it that has some noise.}
    \label{fig:ev-qual-setup}
\end{figure}

The original accepted process execution and the process execution with noise can be seen in \autoref{fig:ev-qual-setup}.
In this example, the activity \textit{Cancel application} is removed from the application trace, as well as the activity \textit{Cancel offer}. 
Instead of \textit{Cancel offer} the activity \textit{Accept offer} was recorded for \textit{application 1}.
In the trace of \textit{offer 1}, nothing was changed.
An error like that could be introduced to a real process by human mistakes.
The change created multiple deviations from the wanted behavior described by the de-jure model.
For example, it is not possible to accept an application without accepting an offer.
This is an inter-object dependency we would want alignments to detect.
Also, the events \textit{Cancel application} and \textit{Cancel offer} are now only recorded for an \textit{offer} but are not connected with an \textit{application}. 
Therefore, this represents unwanted behavior as well.

We applied the traditional alignment approach \cite{adriansyah_aligning_2014} after flattening \cite{aalst_discovering_2020} the object-centric process execution and the object-centric Petri net.
This results in the two alignments.
We also applied our object-centric alignment approach to the described evaluation data.
Both the traditional and object-centric alignment can be seen in \autoref{fig:ev-both-algn}.

\begin{figure}[t]
    \centering
    \includegraphics[width=0.75\textwidth]{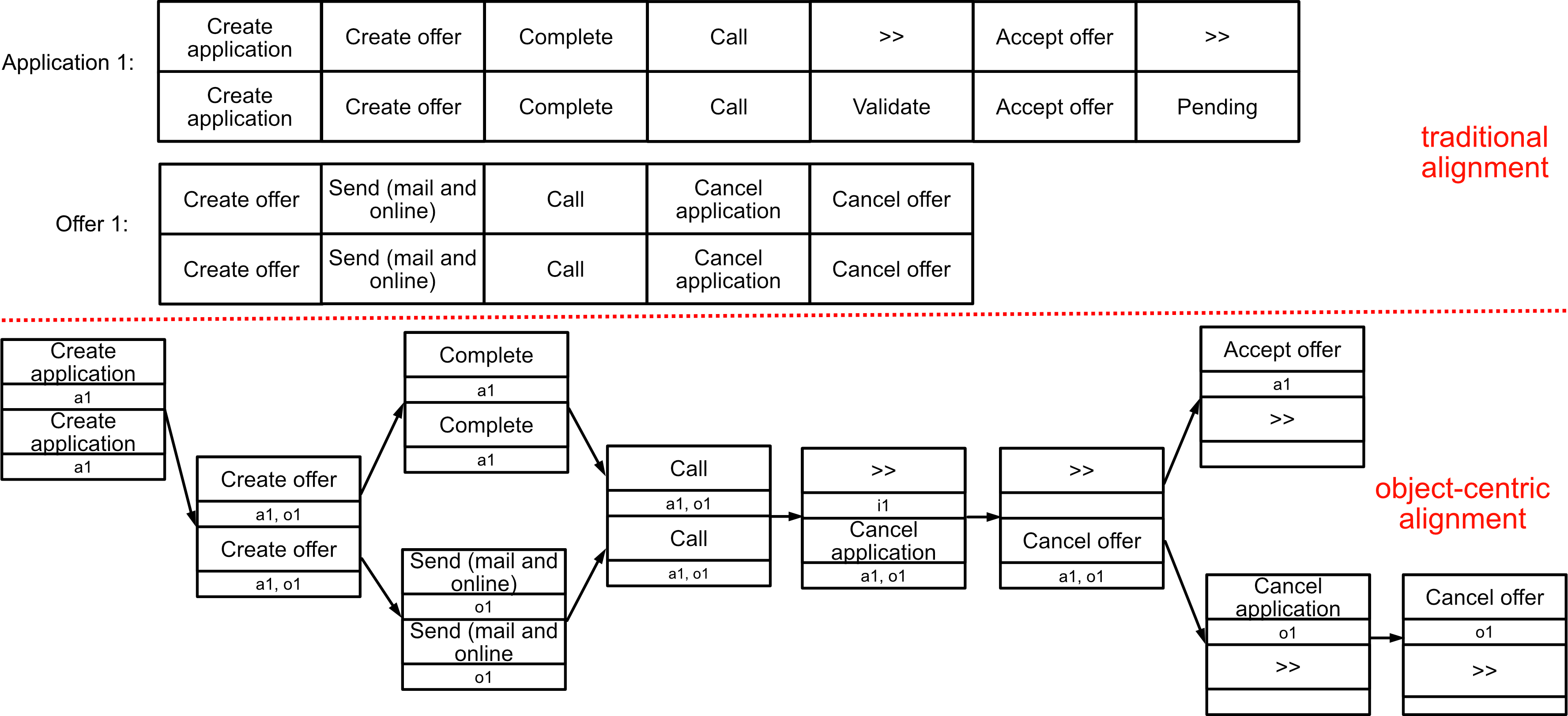}
    \caption{Traditional and object-centric alignments for the variant with noise and the evaluation de-jure model.}
    \label{fig:ev-both-algn}
\end{figure}

Comparing those two, we can see big differences in the diagnostics.
The biggest difference is that the traditional alignment did not detect that \textit{application 1} should not have been accepted.
The reason for that is that with the traditional approach, the alignment for \textit{application 1} is computed without any knowledge about any offers.
Therefore, it has to assume that activity \textit{Accept Offer} was a valid activity, where instead there was no matching offer to accept.
This missing information leads to the traditional alignment suggesting that activities \textit{Validate} and \textit{Pending} are missing in the log, further supporting the understanding that this application should have been accepted.
Our approach considers all dependencies between objects and therefore identifies \textit{Accept offer} as a log move.
Moreover, the traditional alignment does not indicate any control flow deviation for \textit{offer 1}.
This creates a contradiction between the alignment of \textit{offer 1} and \textit{application 1} because an application can not be accepted and canceled.
Contradictions like that can not occur in an object-centric alignment because the whole process execution is considered at once.

\subsection{Quantitative}

To evaluate the scalability of our approach, we performed a quantitative analysis of the run time.

\subsubsection{Evaluation Setup}
As the data source we used BPI2017 event data \cite{van_dongen_bpi_2017}.
Only the most frequent 50\% of activities were used.
All other activities were filtered out.
Afterward, the log consisted of 755 variants.
We used a Petri net designed so that the given log contains some dis-aligned process executions.
The used Petri net is available on GitHub\footnotemark[1].
It has 6 visible transitions and 4 silent transitions.
There are 10 places in the net.
We aligned all the 755 variants with the model and tracked their properties and the resulting alignment calculation time.
The raw results of that evaluation can be found at GitHub\footnotemark[1].
The evaluation was performed on a 3.1 GHz Dual Core Intel Core i5 with 8GB of RAM.

\begin{figure}[h]
    \centering
    \includegraphics[width=0.65\textwidth]{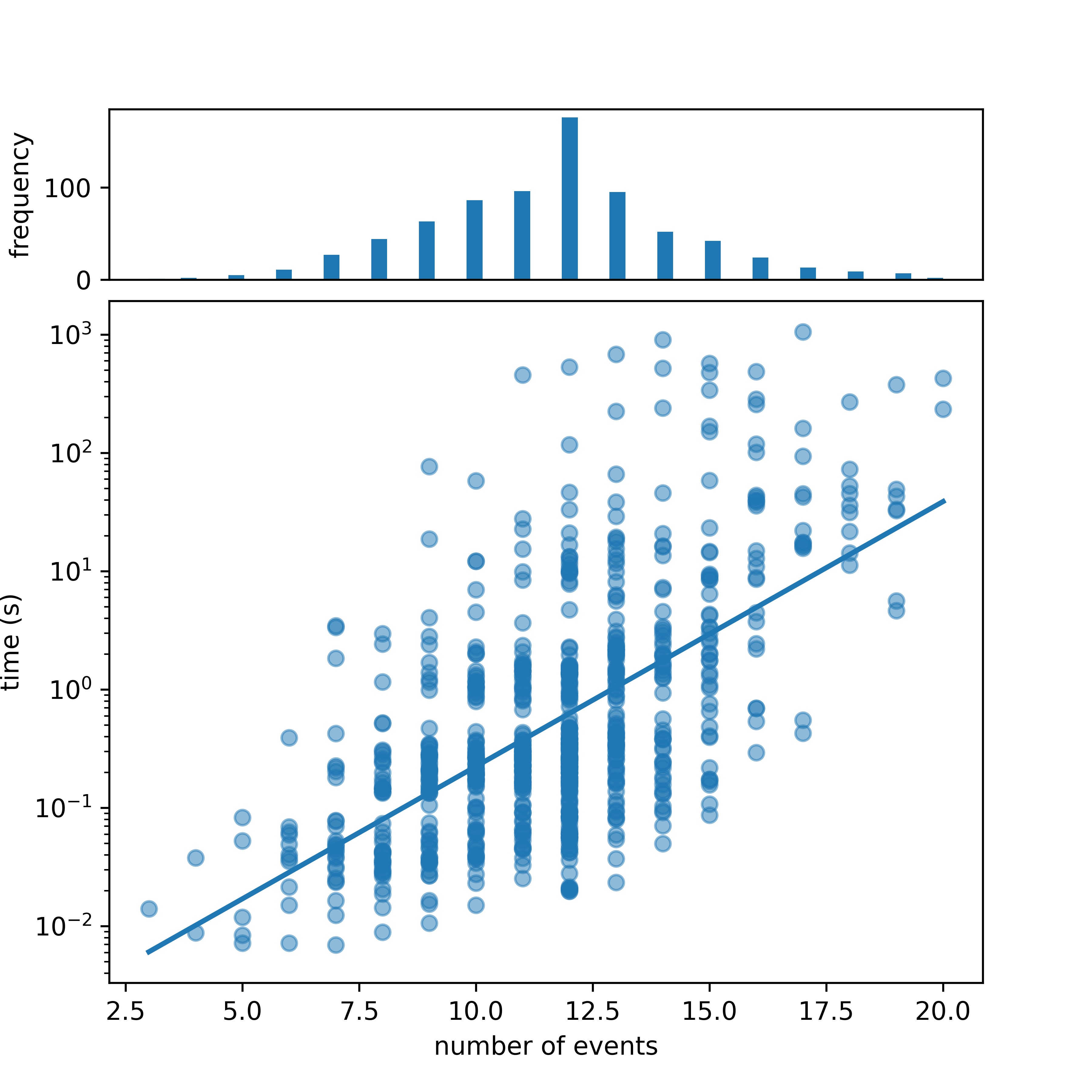}
    \caption{Execution time on a logarithmic scale over the number of events.}
    \label{fig:time-events}
\end{figure}

\subsubsection{Results}

Within the 755 calculated alignments, the cost varied from 0 to a maximum of 5.
The number of events spanned 3 events up to 20 events with the majority of process execution having 11 to 14 events.
The object instances involved started with 2 and went up to 7 instances.
The distribution of these attributes can be seen in \autoref{fig:time-events}, \autoref{fig:time-obj}, and \autoref{fig:ev-time-cost}.
The shortest calculation time was 0.007 seconds and the longest was 1051.8 seconds.
One can see that near the limits of the value range for all attributes, there are fewer data points, which makes the results for these values less resistant to outliers.
We computed the correlation coefficient for each pair of attributes.
The number of events and the number of objects show a positive correlation of 0.59 whereas the number of events and the cost show a negative correlation of -0.38.
Especially the latter one is surprising because one would expect more deviations and therefore a higher cost for process variants with more events.

We plotted the calculation time over the three dimensions \textit{number of events} (\autoref{fig:time-events}), \textit{number of objects} (\autoref{fig:time-obj}), and cost (\autoref{fig:ev-time-cost}).
To reduce the effect of the negative correlation, we grouped the data points in \autoref{fig:ev-time-cost} by their number of events.
Note that the time scale for the time is logarithmic in all three plots and, therefore, the linear increase represents exponential growth.
For all three dimensions, the calculation time grows exponentially.


\begin{figure}[h]
    \centering
    \includegraphics[width=0.65\textwidth]{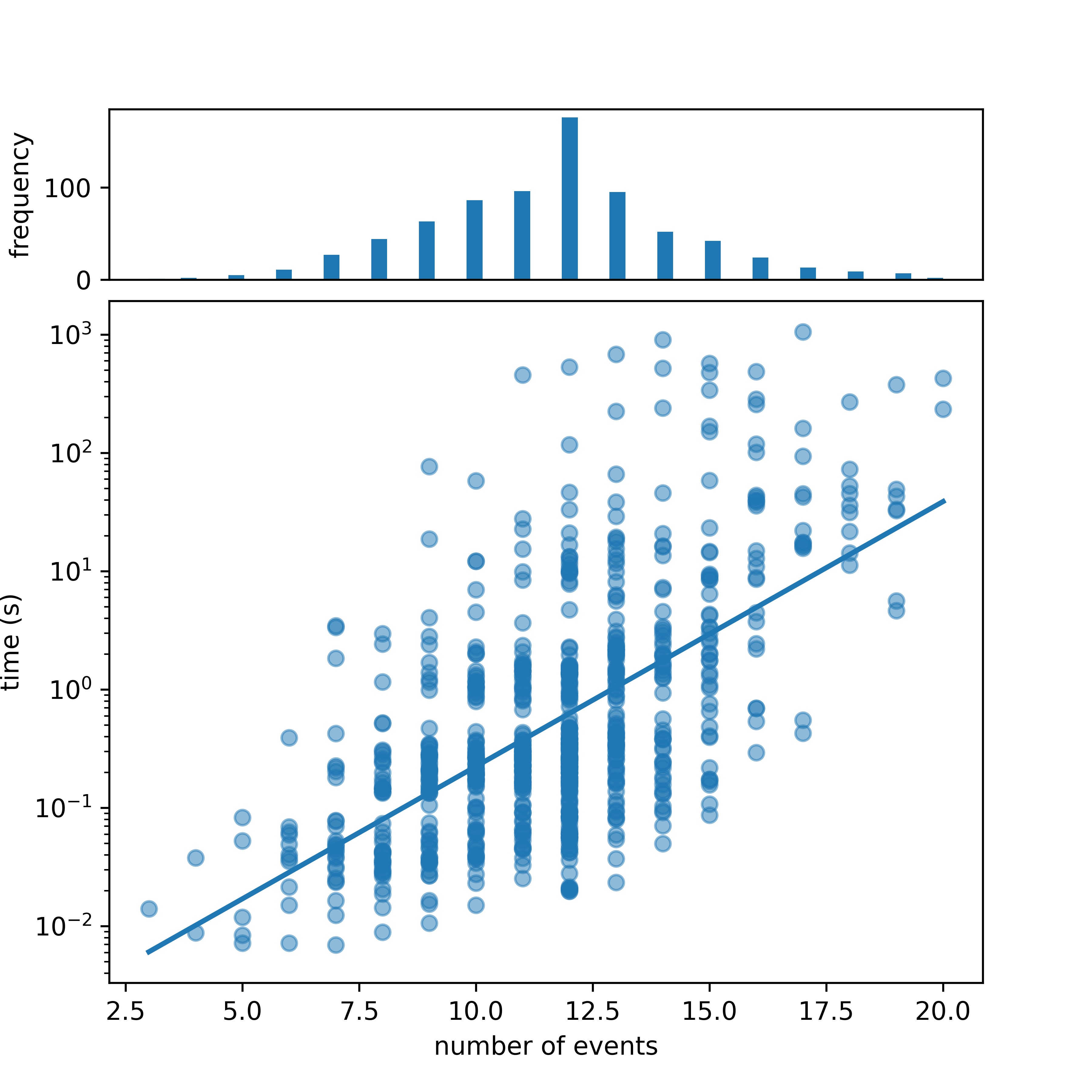}
    \caption{Execution time on a logarithmic scale over the number of objects.}
    \label{fig:time-obj}
\end{figure}

\begin{figure}[h]
    \centering
    \includegraphics[width=0.65\textwidth]{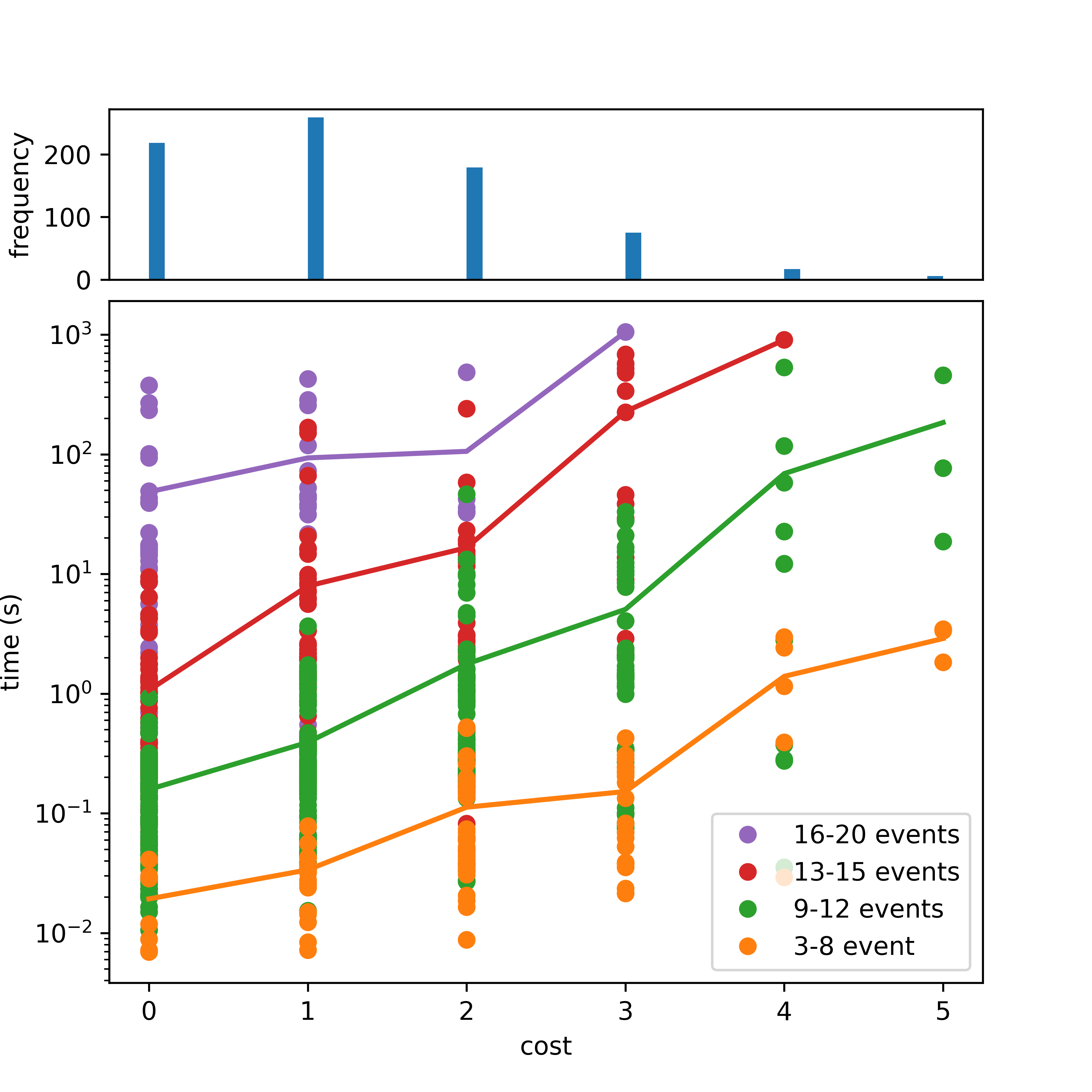}
    \caption{Execution time on a logarithmic scale over alignment cost, grouped by the number of events.}
    \label{fig:ev-time-cost}
\end{figure}

This result can be explained by the structure of the search space.
More events increase the size of the process execution net and can also increase the number of synchronous transitions in the synchronous product net.
Thus there are more possible markings which leads to a bigger search space.
More objects also increase the size of the process execution net and add a lot of parallel behavior, thereby also increasing the search space.
The run time of the Dijkstra algorithm is on average exponential in the size of the search space \cite{dijkstra_note_1959}.
Therefore, the computation time grows exponentially for the number of events and objects.
For an alignment with a higher cost, a bigger portion of the search space is explored because Dijkstra explores all vertices that are reachable by a path that is cheaper than the cheapest path to the final marking.
As a result, the run-time grows exponentially for the cost of the alignment.

\section{Conclusion}
\label{sec:conclusion}
This paper presented four contributions for conformance checking in object-centric process mining.
First, we defined object-centric alignments generalizing traditional alignments to graphs of moves.
Second, we provided an algorithm to calculate optimal alignments on an object-centric log and Petri net.
The two-step approach creates a synchronous product net and searches for the optimal binding sequence from initial marking to final marking.
Third, we implemented this algorithm using the open-source library \textsc{ocpa} \cite{adams_ocpa_2022} and made it publicly available on GitHub\footnotemark[1].
Finally, we performed an evaluation of the presented approach.
The qualitative evaluation shows the advantages of object-centric alignments for deviation diagnostics.
Complex inter-object dependencies are lost when flattening the event data, leading to contradictions in the alignment. Using object-centric alignments, we preserve these dependencies and avoid contradictions.
Our quantitative evaluation indicates an exponential computation time in the number of object instances and cost of the alignment.
This suggests, that an alignment of a whole object-centric event log to a moderately fitting model might be too time-consuming.
In those scenarios, one might use object-centric alignments to get specific diagnostics for individual process executions or variants.

There are two directions for future work based on object-centric alignments.
On the one hand, one can investigate lifting restrictions of the current approach.
For example, the assumption of a fixed object set could be dropped, allowing for the approach to introduce completely new objects that might be missing.
On the other hand, one can work towards decreasing the complexity and run time: Heuristics, using the $A^*$ algorithm, and defining relaxations of the problem are all promising directions to decrease the computation time.

%
%
%
 \bibliographystyle{splncs04}
 \bibliography{bibtex-entries}
%

\end{document}